\documentclass[journal]{IEEEtran}
 \usepackage{amsmath,amssymb}
 \usepackage{subfigure}
 \usepackage{graphicx,graphics,color,psfrag}
 \usepackage{cite,balance}
 \usepackage{caption}
 \allowdisplaybreaks
 \usepackage{algorithm}
 \usepackage{accents}
 \usepackage{amsthm}
 \usepackage{bm}
 \usepackage{algorithmic}
 \usepackage[english]{babel}
 \usepackage{multirow}
 \usepackage{enumerate}
 \usepackage{cases}
 \usepackage{stfloats}
 \usepackage{dsfont}
 \usepackage{color,soul}
 \usepackage{amsfonts}
 \usepackage{cite,graphicx,amsmath,amssymb}
 \usepackage{subfigure}
 \usepackage{fancyhdr}
 \usepackage{hhline}
 \usepackage{graphicx,graphics}
 \usepackage{array,color}
 \usepackage{amsmath}

\newtheorem{theorem}{Theorem}

\newtheorem{definition}{Definition}

\newtheorem{lemma}{Lemma}
\newtheorem{corollary}{Corollary}

\newtheorem{proposition}{Proposition}

\newtheorem{assumption}{Assumption}

\newtheorem{remark}{\bf Remark}

\def\proof{\noindent{\emph{Proof:} }}

\def\E{\mathsf{E}}

\def\phi{\varphi}

\def\SIR{\mathsf{SIR}}

\def\l{\left}
\def\r{\right}
\def\({\left(}
\def\){\right)}

\def\b0{{\mathbf{0}}}

\newcommand{\nn}{\nonumber}

\begin{document}

%%%%%%%%%%%%%
% TITLE
\title{\huge {Mitigating Interference in Content Delivery Networks by Spatial Signal Alignment: The Approach of Shot-Noise Ratio}}

\author{Dongzhu Liu and Kaibin Huang
\thanks{\noindent  D. Liu and K. Huang are with the Department of Electrical and Electronic Engineering, The University of Hong Kong, Pok Fu Lam, Hong Kong (Email: dzliu@eee.hku.hk, haungkb@eee.hku.hk).  The work was supported by Hong Kong Research Grants Council under the Grants 17209917 and 17259416. Part of this work has been presented in IEEE Globecom 2017.
}}
\maketitle

\begin{abstract}
Multimedia content especially videos is expected to dominate data traffic in next-generation mobile networks.  Caching popular content at the network edge, namely content helpers (base stations and access points), has emerged to be a solution for low-latency content delivery.  Compared with the traditional wireless communication, content delivery has a key characteristic that many signals coexisting in the air carry identical popular content.  However, they can interfere with each other at a receiver if their \emph{modulation-and-coding} (MAC) schemes are adapted to individual channels following the classic approach.  To address this issue, we present a novel idea of \emph{content adaptive MAC} (CAMAC) where adapting MAC schemes to content ensures that all signals carry identical content are encoded using an identical MAC scheme to achieve spatial MAC alignment.  Consequently, interference can be harnessed as signals to improve the reliability of wireless delivery.  In the remaining part of the paper, we focus on quantifying the gain that CAMAC can bring to a content-delivery network by using a stochastic-geometry model.  Specifically, content helpers are distributed as a Poisson point process and each of them transmits a file from a content database based on a given popularity distribution.  Given a fixed threshold on the signal-to-interference ratio for successful transmission, it is discovered that the successful content-delivery probability is closely related to the distribution of the ratio of two independent shot noise processes, named a \emph{shot-noise ratio}.  The distribution itself is an open mathematical problem that we tackle in this work.  Using stable-distribution theory and tools from stochastic geometry, the distribution function is derived in closed form.  Extending the result in the context of content-delivery networks with CAMAC yields the content-delivery probability in different closed forms.  In addition, the gain in the probability due to CAMAC is shown to grow with the level of skewness in the content popularity distribution.  
\end{abstract}

%\begin{IEEEkeywords}
%Network edge caching, content delivery, shot-noise process, adaptive modulation, radio access networks.
%\end{IEEEkeywords}

\section{Introduction}
Videos and other types of multimedia data are becoming increasingly dominant in mobile traffic and undergoing exponential growth.  This gives next-generation mobile networks a key mission of supporting low-latency and reliable content delivery.  It is widely agreed that caching popular content at content helpers (e.g. base stations and access points) at the network edge is a promising solution.  In this context, the paper presents a novel algorithm for efficient content delivery, called \emph{content adaptive modulation-and-coding} (CAMAC).  The algorithm aligns the \emph{modulation-and-coding} (MAC) schemes used by content helpers to allow users to retrieve useful content from interference.  Furthermore, we analyze the performance of a content-delivery network adopting CAMAC using a stochastic geometry model.  In the process, an open problem concerning the distribution of the ratio of two shot noise processes is solved.

\subsection{Techniques for Content Caching and Delivery}
Under the constraint that helpers have finite storage, efficient content delivery requires the joint design of the techniques for content caching and delivery.  One approach called coded caching, is to jointly encode content cached at multiple helpers such that the broadcast nature of wireless transmission can be exploited for efficient content delivery by reducing the number of channel usages \cite{maddah2014fundamental, ghorbel2016content}.  Another approach does not involve coding and focus solely on optimizing the content placement at multiple helpers to reduce delivery latency.  Solving the problem of optimal content placement was found to be NP-hard \cite{shanmugam2013femtocaching, li2015distributed}.  Thus, most research based on the current approach aims at designing sub-optimal techniques using diversified tools such as greedy algorithms\cite{shanmugam2013femtocaching} and belief-propagation \cite{li2015distributed}.  A survey of recent advancements in this direction can be found in \cite{paschos2016wireless}.

In this work, we propose a new approach for efficient content delivery from the perspective of adaptive MAC.  In classic communication theory, the MAC scheme is adapted to the channel state for coping with channel fading \cite{goldsmith1997capacity, caire1999capacity}.  In contrast, the proposed CAMAC design is to adapt MAC to the transmitted content file.  The design is motivated by the fact that in a content-delivery system, many coexisting signals in the air carry the same popular content.  Then CAMAC ensures all such signals are encoded using an identical MAC scheme, allowing them to combine at any intended/unintended receivers instead of interfering with each other.  In other words, CAMAC represents a low-cost cooperative scheme for content helpers without message passing between them.  The design effectively coordinates all helpers into a content-broadcast system.

\subsection{Modeling and Designing Content-Delivery Networks}
The design and analysis of large-scale content-delivery networks is currently a highly active area.  The research in the area focuses on designing strategies for caching content at helpers with finite storage so as to optimize the network performance.  Stochastic-geometric network models are commonly adopted in this area where content helpers are typically distributed as \emph{Poisson point process} (PPP).  Correspondingly, the network performance is measured by the product of \emph{hit probability} \cite{blaszczyszyn2015optimal}, defined as the probability that a file requested by a typical user is available at the associated helper and the \emph{content-delivery probability}, defined as the probability that a transmitted file is received by an intended user with a \emph{signal-to-interference ratio} (SIR) or \emph{signal-to-interference-noise ratio} (SINR) exceeding a given threshold \cite{haenggi2009stochastic}.  Based on the network model and performance metrics, caching strategies have been designed for various types of content-delivery networks including device-to-device networks \cite{krishnan2017effect, afshang2016optimal, malak2014optimal,ji2016wireless}, cellular networks \cite{blaszczyszyn2015optimal}, heterogeneous networks \cite{yang2016analysis,bastug2015cache,cui2016analysis, wen2017cache}, and cooperative networks \cite{chen2017cooperative, wen2017random}.  Depending on whether the decision on caching a file at a helper is fixed or randomized caching, the strategies can be separated into deterministic caching \cite{krishnan2017effect, afshang2016optimal, yang2016analysis, bastug2015cache, chen2017cooperative} and random caching \cite{blaszczyszyn2015optimal,malak2014optimal, cui2016analysis,wen2017cache,wen2017random, chen2017probabilistic}.  The common findings of this series of research are that both popular and unpopular content files should be made available in networks but the corresponding helper decisions depend on their popularity as well as the helpers' parameters e.g., storage capacity and transmission power.  

The prior work shares the assumption that the signals transmitted by helpers appear as interference at unassociated users even if they attempt to receive the same content as carried in the interference signals \cite{krishnan2017effect, afshang2016optimal, malak2014optimal, blaszczyszyn2015optimal, yang2016analysis, bastug2015cache, cui2016analysis, wen2017cache}.  The assumption implies that the MAC schemes at non-cooperative helpers are independently adapted, resulting in independently distributed signals regardless of their content.  This justifies the said assumption.  In contrast, the deployment of the proposed CAMAC algorithm in a content-delivery network unifies the MAC schemes adopted by all helpers in the plane that transmit a same content file.  Consequently, all their signals can be combined at any users requesting the file before demodulation and decoding, suppressing interference with respect to the case without MAC alignment.  

\subsection{Signal-and-Interference Distributions in Wireless Networks}\label{sec: first def shot noise ratio}
Due to its tractability and availability of many analytical tools, the theory of PPP has been widely applied in modelling and analyzing wireless networks (see e.g., the survey in \cite{haenggi2009stochastic}).  The network-performance analysis, typically, the analysis of outage probability involves characterizing the distributions of shot noise process representing random signals and interference \cite{haenggi2009stochastic}.  Given a PPP $\Phi$ in the plane, a shot noise process is defined as the following summation over $\Phi$:
\begin{equation}
S\left(\Phi\right)=\l\{
\begin{aligned}
&\sum_{X\in\Phi} h_X \left|X\right|^{-\alpha}, && \text{with fading},\\
&\sum_{X\in\Phi} \left|X\right|^{-\alpha}, && \text{without fading}
\end{aligned}\r.
\label{Eq:Shot:Def}
\end{equation}
where the fading coefficients $\{h_X\}$ are \emph{random variables} (r.v.s) independent of $\Phi$, the path-loss exponent $\alpha$ is a positive constant and $\left|X\right|$ measures the distance between $X$ and the origin.  For networks without cooperative transmissions, the outage probability usually depends on the distribution of a single shot noise process e.g., its Laplace function in the case of Rayleigh fading \cite{andrews2011tractable, baccelli2009stochastic}.  On the other hand, for networks with cooperative transmissions, the outage-probability analysis involves studying the ratio of two shot noise processes, called a \emph{shot-noise ratio}, to model the signal-to-interference ratio \cite{nigam2014coordinated, huang2013analytical, tanbourgi2014tractable, guruacharya2017sinr}.  One relatively simple model of cooperative transmission is to divide the plane into two non-overlapping regions with respect to the typical user: one is the near-field region containing associated transmitters and the other one is far-field region containing interferers \cite{huang2013analytical, tanbourgi2014tractable, guruacharya2017sinr}.  The transmitters and interferers are subsets of a homogeneous PPP.  As a result, the outage-probability analysis in \cite{huang2013analytical, tanbourgi2014tractable, guruacharya2017sinr} reduces to analyzing a shot-noise ratio arising from two \emph{non-homonegeous} PPPs.  Due to its intractability, prior work resorts to asymptotic analysis \cite{huang2013analytical}, approximation \cite{tanbourgi2014tractable}, or presenting  complex expression requiring numerical evaluation  \cite{guruacharya2017sinr}.

In this work, for a content-delivery network adopting CAMAC, it is found that the content-delivery probability depends on the distribution of a shot-noise ratio generated by two \emph{homogeneous} PPPs $\Phi_1$ and $\Phi_2$ with densities $\lambda_1$ and $\lambda_2$, denoted as $R\left(\lambda_1,\lambda_2\right)=\frac{S\left(\Phi_1\right)}{S\left(\Phi_2\right)}$.  The distribution remains unknown in the literature.  We tackle the open problem and derive the distribution of $R\left(\Phi_1, \Phi_2\right)$ by applying the theory of stable distribution as elaborated shortly.

\subsection{Contributions and Organization}

We consider a distributed content-delivery network where content helpers are distributed as a homogeneous PPP in the horizontal plane.  There exists a content database comprising a fixed number of files.  They are requested by a typical user with varying probabilities, forming the \emph{popularity distribution}.  For simplicity, it is assumed that due to limited storage, each helper randomly caches a single file based on the popularity distribution.\footnote{\label{Foot1}{It is straightforward to extend the current analysis  to the case where each helper caches more than one file.  The main effect is that the density of helpers caching content useful for  the typical user is increased by a factor proportional to the total  popularity measure   of the files cached by the helper. This shortens transmission distances between users and associated helpers, thereby increasing the content-delivery probability. In the extreme case where helpers with large storages cache the whole database, each user is served by the nearest helper among all. For the above general cases, the derivation procedure remains unchanged as in the current work.}} In the content-delivery phase, all helpers transmit their cached files to associated users requesting the files.  Our analysis focuses on a typical user at the origin, associated with the nearest helper having the file requested by the user.  We measure the network performance by the content-delivery probability that the typical user receives the desirable file with the received SIR exceeding a given threshold.  

Given the network model, the contributions of the current work are summarized as follows. 

\begin{enumerate}
\item  (CAMAC Algorithm)
The first contribution of the work is the idea of spatial MAC alignment for delivering identical content and its realization by designing the following CAMAC algorithm.  There exists a pre-determined mapping from the files in a given content database to a set of available MAC schemes.  The mapping is known to all helpers and used by them to adapt the MAC scheme to their transmitted files. 
\item (Shot-Noise Ratio)
The second contribution of the work is the tractable analysis of the distribution of a shot-noise ratio and the application of the results to derive the content-delivery probability.  Consider a shot-noise ratio $R(\lambda_1, \lambda_2)$ generated by two independent homogeneous PPPs $\Phi_1$ and $\Phi_2$, with densities $\lambda_1$ and $\lambda_2$ as defined in \eqref{Eq:Shot:Def} of the case without fading.  First we derive the \emph{complementary cumulative distribution function} (CCDF) of $R(\lambda_1, \lambda_2)$ by converting the function to the zero-crossing probability of a \emph{differential shot noise process}, defined as the difference between two shot noise processes, which is proved to be subject to the class of stable distribution.  Then applying the theory of stable distribution leads to the desirable CCDF in a  closed form:
\begin{align}\label{Eq:CCDF:ShotR}
\Pr\left(R\left(\lambda_1,\lambda_2\right)>x\right)= \frac{\alpha}{2\pi}\arctan&\Bigg(\Big(-1+\frac{2}{1+\frac{\lambda_2 }{\lambda_1}x^{\frac{2}{\alpha}}}\Big)\nn\\
&\times \tan\frac{\pi}{\alpha}\Bigg) +\frac{1}{2}. 
\end{align}
Next, by applying Campbell's theorem \cite{martin} and the series form of the shot noise distribution \cite{PowerLawShotNoise}, we derive the Laplace function of the shot-noise ratio in a series form:
\begin{equation}\label{Eq:Lapace:ShotR}
\E\left[e^{-sR(\lambda_1, \lambda_2)}\right] = \sum_{m=1}^{\infty}\frac{(-1)^{m+1}}{\Gamma(1-m\frac{2}{\alpha})}\left(\frac{\lambda_2}{\lambda_1}s^{-\frac{2}{\alpha}}\right)^m.
\end{equation} 

\item (Content-Delivery Probability)
Using transmitted files as random marks, the helper process can be decomposed into a set of homogeneous PPPs, each corresponding to a particular file.  Due to CAMAC, the typical user's data signal combines all incident signals for the PPP marked by the requested file and interference power sums over all PPPs marked by other files.  As a result, the content-delivery probability is a distribution function of a shot-noise ratio representing the received SIR.  Applying the preceding results of shot-noise ratio yields simple forms  for the content-delivery probability.  An approximation of the probability is also derived to yield useful insight relating it with the popularity distribution.
\end{enumerate}

The remainder of the paper is organized as follows. The mathematical model and metric are described in Section~\ref{sec: Mathematical Model}.  The CAMAC design  is proposed in Section~\ref{sec:content adaptive modulation and coding}.  The distribution of shot-noise ratio is provided in Section~\ref{sec: Distribution of a shot-noise ratio} followed by the analysis of  content-delivery probability in Section~\ref{sec: coverage analysis}.  Spatial alignment gain is discussed in Section~\ref{sec: discussion}.  Simulation results are presented in Section~\ref{sec: simulation results} followed by concluding remarks in Section~\ref{sec: Concluding Remarks}. The appendix contains the proofs of propositions and lemmas.

\section{Mathematical Model and Metric}\label{sec: Mathematical Model}

Consider a distributed content-delivery network where content helpers are  modeled as a homogeneous \emph{Poisson point process} (PPP) in the horizontal plane, denoted as $\Phi$ with density $\lambda$.  Let $\mathcal{D}\triangleq\left\{\mathcal{F}_1, \mathcal{F}_2, \cdots , \mathcal{F}_N\right\}$ denote a content database comprising $N$ files and each of them with a uniform size. The popularity distribution for the files are denoted as $\{a_n\}$ with $a_n \in \left[ 0, 1\right] $ corresponding to $\mathcal{F}_n$ and  $\sum_{n=1}^N  a_n =1$. To simplify analysis, it is assumed that the storage of each helper is limited so that it randomly caches a single file based on the popularity distribution, corresponding to the case where helpers are mobiles. Each user randomly generates a request  for a particular file based on the popularity distribution and thus is associated with the nearest helper caching the desirable content. We assume that the network is heavily loaded such that every  helper is  associated with a desired user and transmits the requested file in every time slot. Consider an arbitrary time slot.  Using the transmitted files as marks, the helper process $\Phi$ can be decomposed into $N$ homogeneous PPPs $\Phi_1, \Phi_2, \cdots, \Phi_N$ with $\Phi_n$ corresponds to $\mathcal{F}_n$ and having the density $a_n\lambda$. For analyzing the content-delivery probability, based on  Slivnyak's Theorem \cite{chiu2013stochastic}, it is sufficient to consider a typical user at the origin with the request denoted as $D_0\in \mathcal{D}$. The network spatial distribution is illustrated in Fig.~\ref{Fig: spatial}. 

\begin{figure}[t]
\begin{center}
{\includegraphics[width=9cm]{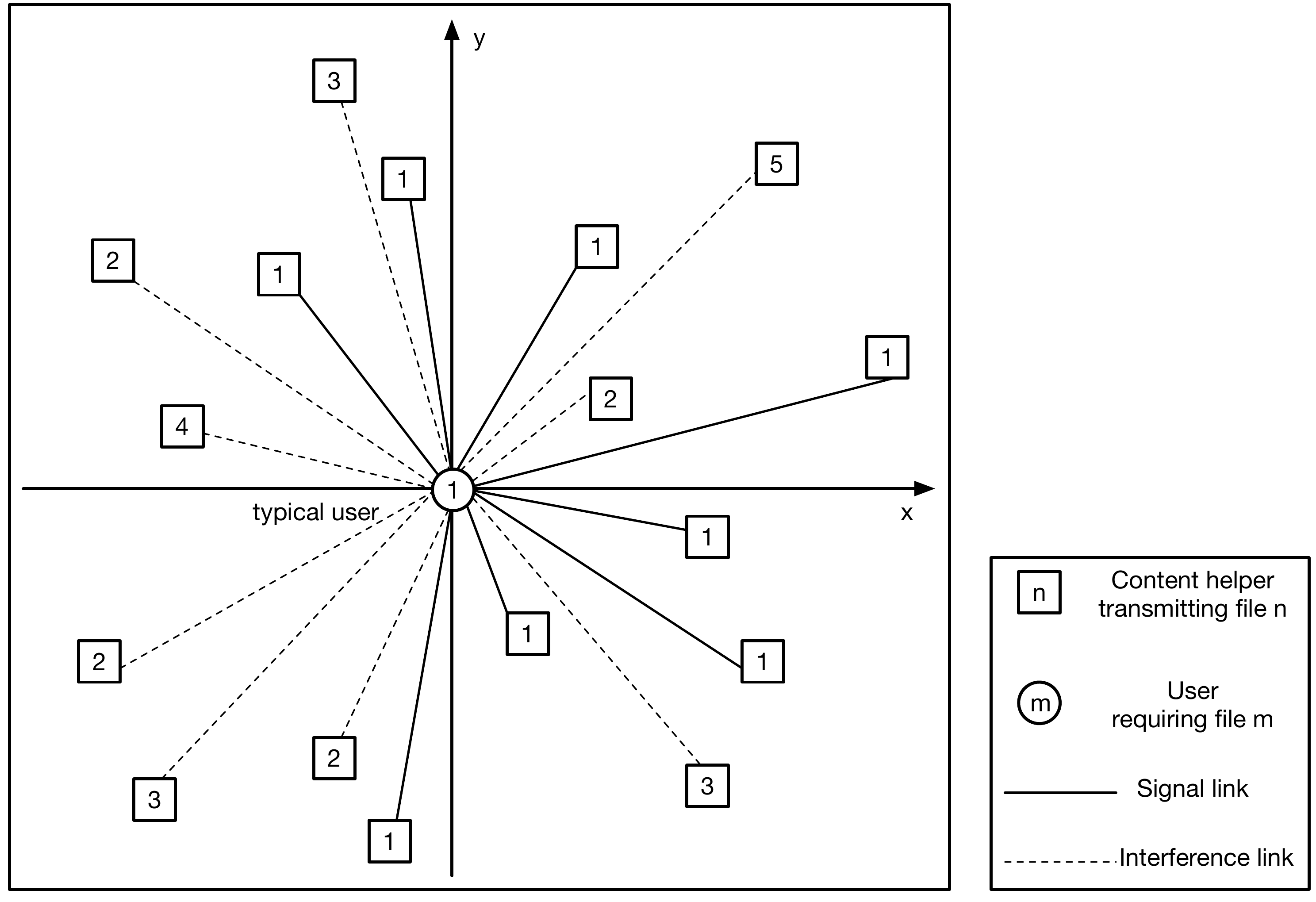}}\\
\caption{Spatial distribution of a content delivery network adapting CAMAC (see Section~\ref{sec:content adaptive modulation and coding}).}
\label{Fig: spatial}
\end{center}
\end{figure}

The channel model is described as follows. All channels are assumed to be narrow-band and characterized by path loss and Rayleigh fading.  Consider an arbitrary time slot. The signal transmitted by a transmitter at the location $X$, denoted as $Q_X$,  is received by the typical receiver as $H_X \left|X\right|^{-\frac{\alpha}{2}} Q_X$, where $\left|X\right|$ measures the Euclidean propagation distance, $\alpha$ denotes the fixed path-loss exponent, and the fading coefficient   $H_X$ is a ${\cal{CN}}(0,1)$ r.v. All fading coefficients are i.i.d.   Consider an interference limited network where noise is negligible. Then  the total received signal at the typical user can be written as: 
\begin{equation}\label{Eq:Signal}
Y_0= \sum_{n=1}^N \sum_{X\in\Phi_n } {H_X} \left| X \right|^{-\frac{\alpha}{2}} Q_X. 
\end{equation}
The signal is decomposed into data signal and interference in the next section based on the  CAMAC algorithm. Let $\SIR$ denotes the resultant SIR for the typical user. Conditioned on $D_0 = \mathcal{F}_k$, the conditional content-delivery probability is denoted as $P_d(\mathcal{F}_k)$ and written as $P_d (\mathcal{F}_k)= \Pr(\SIR > \theta_k\mid D_0 = \mathcal{F}_k)$, with $\theta_k$ is a given threshold specifying the criterion on the received signal  for successful delivery of $\mathcal{F}_k$.   Then the content-delivery probability   $P_d = \sum_{k=1}^N a_k P_d(\mathcal{F}_k)$.

\section{Content Adaptive Modulation and Coding}\label{sec:content adaptive modulation and coding}

\subsection{CAMAC: Algorithm Design} \label{subsec:CAMAC}

The design of the CAMAC algorithm is simple and comprises two components: one is a mapping from files in a content database to a given number of MAC schemes and the other one is the helper architecture for adapting the MAC scheme using the mapping, which are illustrated in Fig.~\ref{Fig: CAMAC}. Consider the content-MAC mapping in Fig.~\ref{Fig: Mapping}. The assignment of a particular MAC scheme to a given file depends on the the acceptable MAC rate given the quality requirement and the size of the content. Typically, the set of available MAC schemes is much smaller than the content database and thus each MAC scheme can be assigned to multiple files.  The optimal design of the mapping depends on specific content characteristics and system parameters, which is outside the scope of the current work.  Next, consider the helper architecture for implementing CAMAC as shown in Fig.~\ref{Fig: Typical Helper}. All helpers agree on an uniform  content-MAC mapping which is determined by a centralized control station prior to the content-delivery phase. During content delivery, each helper uses the mapping to select a MAC scheme for modulating and coding a particular transmitted file. This spatially aligns the MAC schemes of all helpers transmitting identical content without online coordination or message exchange.  As a result, as illustrated in Fig.~\ref{Fig: spatial}, the signals transmitted by all helpers for sending identical files are combined at any receiver requesting the file as the input to the  demodulator and decoder, increasing signal power and reducing interference. 

\begin{figure}[t]
\begin{center}
\subfigure[Mapping between content and MAC schemes.]
{\includegraphics[width=6cm]{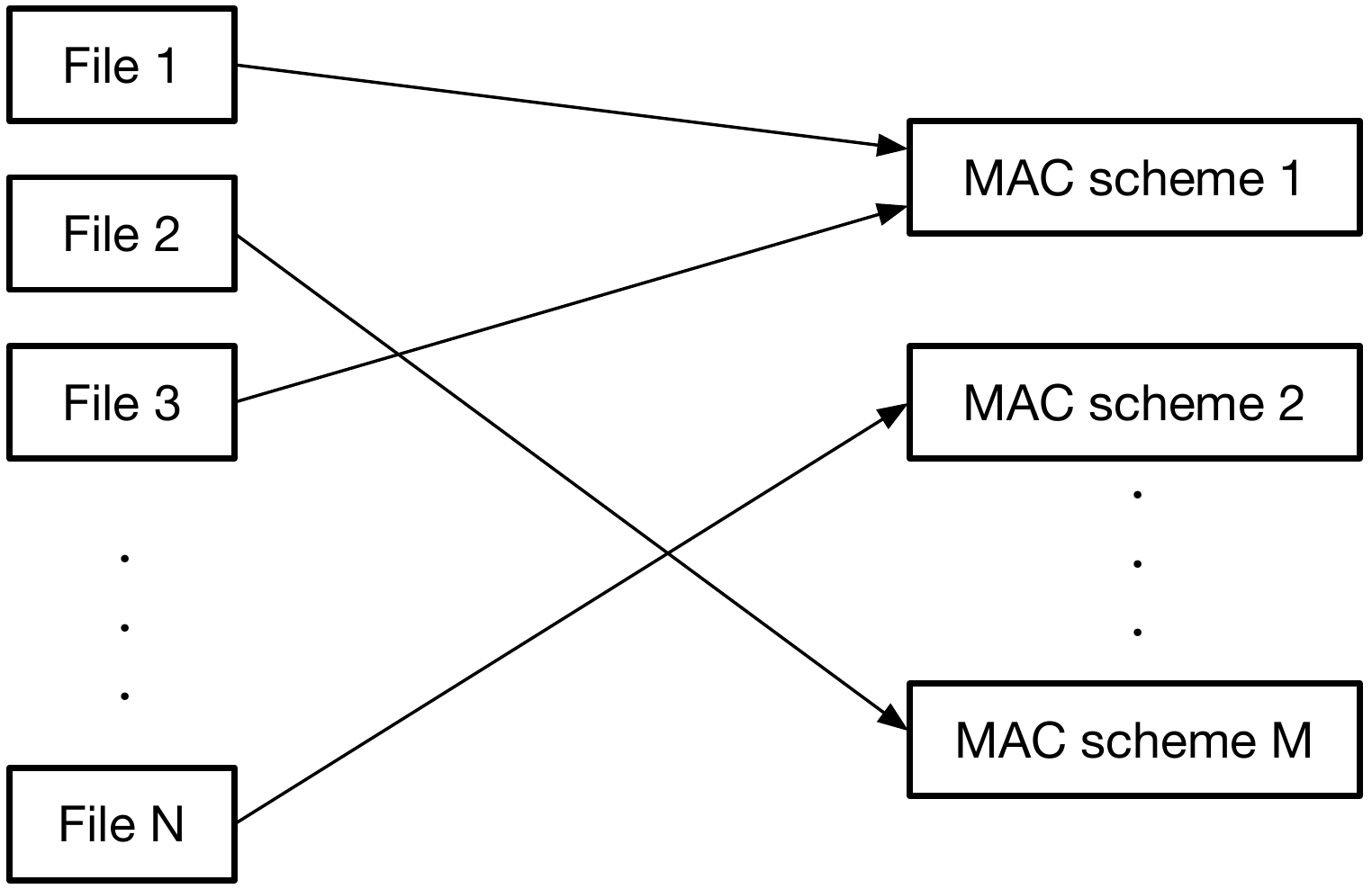}\label{Fig: Mapping}} \\ 
\subfigure[Architecture of a content helper.]
{\includegraphics[width=9cm]{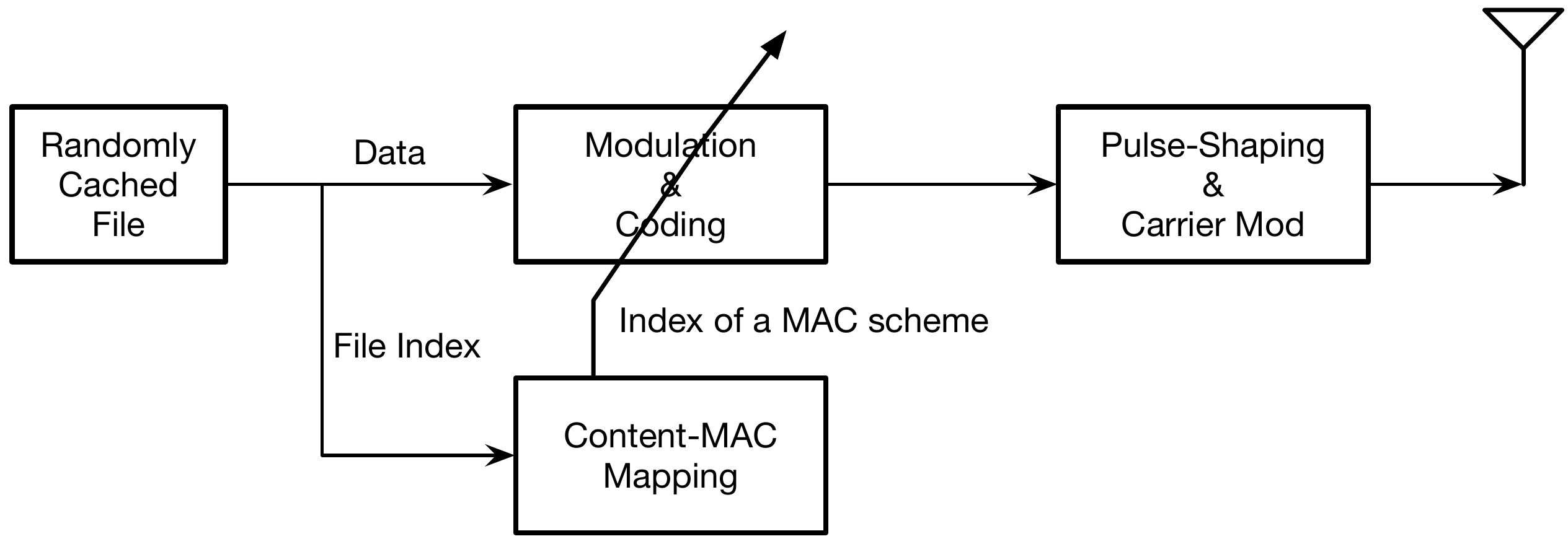}\label{Fig: Typical Helper}}\\
\caption{The design of content adaptive modulation and coding (CAMAC).}
\label{Fig: CAMAC}
\end{center}
\end{figure}

For tractability, two assumptions related to CAMAC are made for the network performance analysis in the sequel. 
\begin{assumption}[Distribution of Transmitted Signals]\label{assump: distribution of transmitted signals}\emph{{The signals transmitted by helpers for sending identical files are identical. Any two transmitted  signals carrying different files are independent r.v.s. 
}}
\end{assumption}

\begin{assumption}[Effective Frequency-Flat Channel]\label{assump: Effective Frequency-Flat Channel}\emph{The effective multi-path channel (see Fig.~\ref{Fig: spatial}) from helpers sending a particular  file to the typical user who needs the file is \emph{frequency flat}. 
}
\end{assumption}
The assumption corresponds to the case where the path-loss exponent is relatively high such that the typical user receives only significant signals from nearby helpers. Consequently, the differentiation in the propagation delay in the signal paths is small and does not induce frequency selectivity. Otherwise, channel equalization is needed at the typical user prior to detection.

\subsection{CAMAC: Signal and Interference}\label{subsec: CAMAC signal and interference}

Given the CAMAC algorithm, the expressions for signal and interference and their distributions are obtained shortly to facilitate the subsequent network-performance analysis. 

Consider the case where the typical user requests $\mathcal{F}_k$. Given CAMAC and Assumptions~\ref{assump: distribution of transmitted signals} and~\ref{assump: Effective Frequency-Flat Channel}, conditioned on $D_0=\mathcal{F}_k$, the signal can be rewritten from \eqref{Eq:Signal} as 
\begin{equation*}
Y_0(\mathcal{F}_k)\!\!=\!\! \l(\sum_{X\in\Phi_k } {H_X} \left| X \right|^{-\frac{\alpha}{2}}\r)\!\! Q_k \!\!+\!\!\sum_{n\neq k} \sum_{Z\in\Phi_n} \!\! \l({H_Z} \left| Z \right|^{-\frac{\alpha}{2}}\r) \!\!Q_n
\end{equation*}
where the first and second terms correspond to (data) signal and interference, respectively. Then the signal and interference power, denoted as $S_0(\mathcal{F}_k)$ and $I_0(\mathcal{F}_k)$, respectively, can be written as
\begin{equation*}
S_0(\mathcal{F}_k)\!\!=\!\!\left|\sum_{X\in\Phi_k }\!\! {H_X} \left| X \right|^{-\frac{\alpha}{2}}\right|^2\!\!,\ I_0(\mathcal{F}_k)\!\!=\!\!\sum_{n\neq k} \left|\sum_{Z\in\Phi_n} \!\! {H_Z} \left| Z \right|^{-\frac{\alpha}{2}}\right|^2\!\!.
\end{equation*}
Since $\{H_X, H_Z\}$ are i.i.d. complex Gaussian r.v.s, the terms $\sum_{X\in\Phi_k } H_X \left| X \right|^{-\frac{\alpha}{2}}$ and $\sum_{Z\in\Phi_n } H_Z \left| Z \right|^{-\frac{\alpha}{2}}$ are distributed as ${\cal{CN}}(0,\sum_{X\in\Phi_k }  \left| X \right|^{-\alpha})$ and ${\cal{CN}}(0,\sum_{Z\in\Phi_n }  \left| Z \right|^{-\alpha})$, respectively.  It follows that the signal power $S_0(\mathcal{F}_k)$ follows the  exponential distribution with mean $\sum_{X\in\Phi_k }  \left| X \right|^{-\alpha}$. Moreover, the  interference power can be written as  $I_0(\mathcal{F}_k)=\sum_{n\neq k } I_n$ where $I_n$ also follows the exponential distribution but with mean $\sum_{Z\in\Phi_n }  \left| Z \right|^{-\alpha}$.  Equivalently,  
\begin{equation}\label{eq: def power with fading}
S_0(\mathcal{F}_k)\sim h_k\sum_{X\in\Phi_k }\left| X \right|^{-\alpha}, \quad I_n\sim h_n\sum_{Z\in\Phi_n }\left| Z \right|^{-\alpha}
\end{equation}
where the operator $\sim$ denotes the equivalence in distribution, and  $\{h_k, h_n\}$ are i.i.d. exponential r.v.s with unit mean.  {Consequently, it is obvious that $S_0(\mathcal{F}_k), \{I_n\}$ are independent r.v.s.}  Using these results, the conditional content-delivery probability defined in Section~\ref{sec: Mathematical Model} can be rewritten as 
\begin{align}\label{Eq: conditional_coverage_def cache one} 
{P}_{d}(\mathcal{F}_k)=\Pr \left(\frac{S_0(\mathcal{F}_k)}{\sum_{n\neq k } I_n} > \theta_k \right)
\end{align}
where $S_0(\mathcal{F}_k)$ and $I_n$ are distributed as in \eqref{eq: def power with fading}. 

{
\begin{remark}[Effects of MAC Schemes] \label{remark: Effects of MAC Schemes}
\emph{The effect of the proposed CAMAC scheme on the analysis of content-delivery probability arises from the simple  fact that two signals carrying the same file and modulated/encoded using the same MAC scheme results in the combining of two channels in the received signal [see (5)]. On the other hand, even if modulated/encoded using the same MAC scheme, two signals carrying two different files cannot lead to the said channel combining. From this aspect, the number of available MAC schemes for transmission does not affect the content-delivery probability. However, in practice, it is necessary to have multiple  MAC schemes to support numerous available date rates, thereby satisfying diversified quality of service (QoS) requirements for different types of content.  In the current model, the QoS requirements are reflected in the SIR thresholds $\{\theta_k\}$.
}
\end{remark}}

{
\begin{remark}[Spatial Alignment in Practice] \label{remark: Spatial Alignment in Practice}
\emph{In practice, the users' requests for an identical file are generated asynchronously but the proposed CAMAC scheme requires synchronization of the transmissions by different helpers.  Based on the CAMAC scheme, the helper synchronization is minimum, namely synchronization of slot boundaries, but their modulation adaptations are independent. In the scenario where a  content file is  small (e.g., a song or the  preview of a video) and thus its  transmission duration is short, forcing the said synchronization  causes delay no more than half the duration, which is compensated by improved reliability in content delivery due to CAMAC. On the other hand, in the scenario where a content file is large (e.g., a movie), the \emph{proactive pushing} technology (see e.g., \cite{zhou2015greendelivery}) can be deployed together with CAMAC to proactively  delivery the file by different helpers to the caches of associated mobile users based on prediction of their future demands  using  information on past requests. In this case, the helper synchronization does not incur any additional latency as experienced by mobile users while the link reliability is enhanced by CAMAC. }
\end{remark}}

\section{Distribution of shot-noise ratio}\label{sec: Distribution of a shot-noise ratio}
One can see from \eqref{Eq: conditional_coverage_def cache one} the content-delivery probability is closely related to a process of shot-noise ratio. For the purpose of deriving the probability, we analyze the distribution of the shot-noise ratio $R(\lambda_1, \lambda_2)$ defined in Section~\ref{sec: first def shot noise ratio}. Specifically, in this section, the CCDF  in \eqref{Eq:CCDF:ShotR} and the Laplace function in \eqref{Eq:Lapace:ShotR} are derived. For ease of notation, they are represented by  $\bar{F}\left(x; \lambda_1, \lambda_2\right)$ with $x> 0$ and $\mathcal{L}_{R\left(\lambda_1,\lambda_2\right)}(s)$ with any complex $s$ for which the result converges, respectively. 

The derivation approach is as follows. We define a new process called a  \emph{differential shot noise process} as the difference between two independent shot noise process $S(\Phi_1)$ and $S(\Phi_2)$ as defined in \eqref{Eq:Shot:Def} without fading.  Mathematically, given a constant $x\in \mathbb{R}^+$, the differential shot noise process, denoted as $M(x; \lambda_1, \lambda_2)$, is defined as 
\begin{equation}
M(x; \lambda_1, \lambda_2) =\sum_{X \in \Phi_1}|X|^{-\alpha}-x\sum_{Z \in \Phi_2}|Z|^{-\alpha}.\label{Eq:DiffShot:Def}
\end{equation} 
Then the CCDF of $R(\lambda_1, \lambda_2)$ is equivalent to the zero-crossing probability of $M(x; \lambda_1, \lambda_2)$: 
\begin{equation}\label{Eq: shot noise minus shot noise}
\boxed {\bar{F}\left(x; \lambda_1, \lambda_2\right)=\Pr\left(M(x)>0\right)}. 
\end{equation}
It is proved in  Section~\ref{Section:DiffShot} that the distribution of $M(x; \lambda_1, \lambda_2)$ belongs to the class of \emph{stable distribution}. Then using this factor and by exploiting the equivalence  in \eqref{Eq: shot noise minus shot noise}, the theory of stable distribution is applied in Section~\ref{Section:ShotNoise} to derive the  desired distribution functions for the shot-noise ratio $R(\lambda_1, \lambda_2)$.

%% subsection CF of shot noise - shot noise
\subsection{ Distribution of a Differential Shot-Noise Process}\label{Section:DiffShot}
In this section, the characteristic function for the differential shot noise process $M(x; \lambda_1, \lambda_2)$ is first derived. Comparing the function with that for the class of stable distribution leads to the conclusion that the process belongs to the class. The details are as follows. 

As the concept of stable distribution is repeatedly used in this and next sub-sections, it is useful to provide the definition (see e.g.,  \cite{nolan:StableDistribution})  as follows.  

\begin{definition}[Stable Distribution] \label{Def: stable distribution}
\emph{
A r.v. $X$ is \emph{stable} in the broad sense   if for any positive constants $a$ and $b$, there exist some positive $c$ and some $d \in \mathbb{R}$ such that 
\begin{equation}\label{eq: stable distribution Def}
aX_1+bX_2=cX+d
\end{equation}
where $X_1$ and $X_2$ are independent and identically distributed  r.v.s having  the same distribution as $X$.  If \eqref{eq: stable distribution Def} holds with $d = 0$, then $X$ is strictly stable or stable in the narrow sense.
}
\end{definition}

The definition implies that the distribution shape of $X$ is preserved under a  linear operation with positive parameters. The characteristic function and some useful properties of the class of stable distribution are provided in Appendix~\ref{app: stable distribution}.

It is well known that a shot noise process $S(\Phi)$ with $\alpha > 2$ belongs to the class of stable distribution \cite{PowerLawShotNoise}. However, it is unknown that if the distribution of $M(x; \lambda_1, \lambda_2)$ is also stable. The answer is affirmative as shown in Proposition~\ref{pro: CF marked point process}. To prove the result, first, the characteristic function of $M(x; \lambda_1, \lambda_2)$, defined as $G(t; \lambda_1, \lambda_2) = \E\l[e^{jtM(x; \lambda_1, \lambda_2)}\r]$ with $t\in \mathds{R}$, is derived as follows.

\begin{lemma}[Characteristic Function of Differential Shot Noise]\label{pro: CF marked point process} \emph{Given $\alpha>2$, the  characteristic function of the differential shot noise process $M(x; \lambda_1, \lambda_2)$  satisfies 
\begin{align} 
\log G\left(t\right)=-&\left(\lambda_1+\lambda_2 x^{\frac{2}{\alpha}}\right) \pi \Gamma\left(1-\frac{2}{\alpha}\right) \cos\frac{\pi}{\alpha}\ \left|t \right|^{\frac{2}{\alpha}}\nn\\
&\times \left(1-j \mathrm{sgn}(t)\frac{\lambda_1-\lambda_2 x^{\frac{2}{\alpha}}}{\lambda_1+\lambda_2 x^{\frac{2}{\alpha}}}\tan\frac{\pi}{\alpha} \right)
\end{align} 
where the sign function $\mathrm{sgn}(t)$ follows the definition in  \eqref{eq: sgn}. 
 }
\end{lemma}
\proof See Appendix~\ref{app: CF marked point process}.

Comparing the derived characteristic function with the Form-A characteristic function for stable distribution given in Lemma~\ref{lemma: CF stable distribution} in Appendix~\ref{app: stable distribution} yields the following result. 

\begin{proposition}[Stable Distribution of Differential Shot Noise]\label{pro: marked PP is stable distribution}\emph{The process of differential shot noise, $M(x; \lambda_1, \lambda_2)$ belongs to the  the class of stable distribution:
\begin{equation}
M(x; \lambda_1, \lambda_2)  \sim  S_A(\delta, \beta_A, 0, \mu_A)
\end{equation}
where $S_A$ represents the class of stable distribution in Form A as defined in Appendix~\ref{app: stable distribution} and the parameters are given as: 
\begin{align}
\delta &= \frac{2}{\alpha}, \nn\\
 \beta_A &= \frac{\lambda_1-\lambda_2 x^{\frac{2}{\alpha}}}{\lambda_1+\lambda_2 x^{\frac{2}{\alpha}}},\nn\\
 \mu_A & =  \left(\lambda_1+\lambda_2 x^{\frac{2}{\alpha}}\right)\pi \Gamma \left(1-\frac{2}{\alpha}\right) \cos\frac{\pi}{\alpha}.
\end{align}
}
\end{proposition}

By setting $x = 0$, the stable distribution of $M(0; \lambda_1, \lambda_2)$ reduces to a shot noise process, $S_A\left(\frac{2}{\alpha}, 1, 0, \lambda_1 \pi \Gamma\left(1-\frac{2}{\alpha}\right)\right)$ as mentioned in \cite{PowerLawShotNoise}.

\subsection{Distribution of a Shot-Noise Ratio}\label{Section:ShotNoise}

\subsubsection{CCDF of a shot-noise ratio}The CCDF of a shot-noise ratio is derived by exploiting the equivalence in \eqref{Eq: shot noise minus shot noise} and applying the stable-distribution properties of  the differential shot-noise process $M(x; \lambda_1, \lambda_2)$ as discussed in the preceding sub-section.

As shown in Proposition~\ref{pro: marked PP is stable distribution}, $M(x; \lambda_1, \lambda_2)$ has the stable distribution $S_A(\delta, \beta_A, 0, \mu_A)$.  To apply the properties of Form-B stable distribution with a normalized parameter $\mu_B=1$ in Lemma~\ref{lemma: Special cases of stable distribution} in Appendix~\ref{app: stable distribution}, the distribution of $M(x; \lambda_1, \lambda_2)$ can be converted into Form B by using the parametric relations in \eqref{eq: relation equal to 1} and \eqref{eq: Relation Form A and B delta neq 1} in Appendix~\ref{app: stable distribution}. Moreover, it is necessary to scaling the process $M(x; \lambda_1, \lambda_2)$ to have a normalized parameter $\mu_B= 1$.  Let the resultant process be denoted as $\widetilde{M}(x; \lambda_1, \lambda_2)$.  Its distribution is derived as shown in the following lemma. 
\begin{lemma}[Normalized Differential Shot Noise]\label{lemma: transfer differential shot noise}
\emph{The distribution of the normalize shot-noise process, $\widetilde{M}(x; \lambda_1, \lambda_2)$, and its relation with that of the original process $M(x; \lambda_1, \lambda_2)$ are given as follows: 
\begin{align}
\widetilde{M}(x; \lambda_1, \lambda_2) &\sim S_B\l(\frac{2}{\alpha}, \beta_B, 0, 1\r)\label{Eq:Sim:1}\\
M(x; \lambda_1, \lambda_2) &\sim{ \mu_B}^{\frac{\alpha}{2}} \widetilde{M}(x; \lambda_1, \lambda_2) \label{Eq:Sim:2}
\end{align}
where $S_B$ denotes Form B of stable distribution as specified in Lemma~\ref{lemma: CF stable distribution} and the parameters are given as: 
\begin{align}
\beta_B &= \frac{\alpha}{\pi}\arctan\left(\frac{\lambda_1-\lambda_2 x^{\frac{2}{\alpha}}}{\lambda_1+\lambda_2 x^{\frac{2}{\alpha}}} \tan\frac{\pi}{\alpha}\right), \nn\\
\mu_B  &= \frac{ \left(\lambda_1+\lambda_2 x^{\frac{2}{\alpha}}\right)\pi \Gamma \left(1-\frac{2}{\alpha}\right) \cos\frac{\pi}{\alpha}}{\cos\frac{\pi \beta_B}{\alpha}}.
\end{align}} 
\end{lemma}
The proof is given in Appendix~\ref{app: transfer differential shot noise}. It follows from Lemma~\ref{lemma: transfer differential shot noise} that the zero-crossing probability for $M(x; \lambda_1, \lambda_2)$ is equal to that for the normalized counterpart $\widetilde{M}(x; \lambda_1, \lambda_2)$, which is given in Lemma~\ref{lemma: Special cases of stable distribution} in Appendix~\ref{app: stable distribution}. Combining this result and the equivalence in \eqref{Eq: shot noise minus shot noise} yields the CCDF for the shot-noise ratio as elaborated below. 

\begin{proposition}[CCDF of a Shot-Noise Ratio]\label{pro: CCDF ratio of shot noise}\emph{Given $\alpha>2$, the CCDF of the process of shot-noise ratio, $R(\lambda_1,\lambda_2)$,  is given as 
\begin{align}\label{eq:ccdf ratio of shot noise}
%\bar{F}\left(x; \lambda_1, \lambda_2\right)&=\frac{1}{2}+\frac{\alpha}{2\pi}\arctan\left(\frac{\lambda_1-\lambda_2 x^{\frac{2}{\alpha}}}{\lambda_1+\lambda_2 x^{\frac{2}{\alpha}}}\tan\frac{\pi}{\alpha}\right). 
\bar{F}\left(x; \lambda_1, \lambda_2\right)\!=\!\frac{\alpha}{2\pi}\arctan\left(\Big(\!-1\!+\frac{2}{1+\frac{\lambda_2 }{\lambda_1}x^{\frac{2}{\alpha}}}\Big)\tan\frac{\pi}{\alpha}\right) +\frac{1}{2}. 
\end{align}
}
\end{proposition}
Since $\arctan$ is a monotone increasing function, we can see from the above result that $\bar{F}\left(x; \lambda_1, \lambda_2\right)$ monotonically decreases with the growing ratio  ${\lambda_2}/{\lambda_1}$ besides  $x$. In other words, the function does not depend on the exact values of the densities. For the extreme cases of $x = 0$ and $x\rightarrow \infty$, the expression for $\bar{F}\left(x; \lambda_1, \lambda_2\right)$ reduces to $1$ and $0$, respectively, since the shot-noise ratio is always positive and a proper r.v. 

\subsubsection{Laplace function  of a shot-noise ratio} It is straightforward to derive the Laplace function of a shot-noise process $S(\Phi)$ in \eqref{Eq:Shot:Def} by applying Campbell's Theorem (see e.g.,\cite{kingman1993poisson}). This is not true for the shot-noise ratio, $R\left(\lambda_1,\lambda_2\right)$, mainly due to a shot-noise process in its denominator. Specifically, applying Campbell's Theorem reduces the two shot-noise processes in the characteristic function of  $R\left(\lambda_1,\lambda_2\right)$ to be only one in denominator: 
\begin{align}
\mathcal{L}_{R\left(\lambda_1,\lambda_2\right)}(s)&=\E\left[\prod_{X\in \Phi_1}e^{-\frac{s\left|X\right|^{-\alpha}}{\sum_{Z\in \Phi_2}\left|Z\right|^{-\alpha}}}\right] \nn \\
&=\E\Bigg[e^{-\lambda_1 \pi \Gamma\left(1-\frac{2}{\alpha}\right)\left(\frac{s}{\sum_{Z\in \Phi_2}\left|Z\right|^{-\alpha}}\right)^{\frac{2}{\alpha}}}\Bigg].\label{eq: mid LT shot-noise ratio}
\end{align}
The current approach for deriving the function in closed form from \eqref{eq: mid LT shot-noise ratio} relies on applying  the following expansion of the \emph{probability distribution function} (PDF) of a shot-noise process \cite{PowerLawShotNoise}: 
\begin{align}\label{pdf: shot noise}
f_{S(\Phi)}(x)=\frac{1}{\pi x }\sum_{m=1}^{\infty}&\frac{(-1)^{m+1}\Gamma(1+m\frac{2}{\alpha})\sin \pi m \frac{2}{\alpha}}{m !}\nn\\
&\times \left(\frac{\lambda \pi \Gamma\left(1-\frac{2}{\alpha}\right)}{x^{\frac{2}{\alpha}}}\right)^m.
\end{align}
The result is given in the following proposition. 
\begin{proposition}[Laplace Function of a Shot-Noise Ratio]\label{pro: LT ratio of shot noise process} \emph{Consider a shot-noise ratio $R\left(\lambda_1,\lambda_2\right)$ generated by  two independent homogeneous PPPs $\Phi_1$ and $\Phi_2$ with density $\lambda_1$ and $\lambda_2$ separately.  The  Laplace function  of $R\left(\lambda_1, \lambda_2\right)$ can be written in the following series form: 
\begin{align}\label{eq: LT ratio of shot noise process}
\mathcal{L}_{R\left(\lambda_1,\lambda_2\right)}\left(s\right)=\sum_{m=1}^{\infty}\frac{(-1)^{m+1}}{\Gamma(1-m\frac{2}{\alpha})}\left(\frac{\lambda_2}{\lambda_1}s^{-\frac{2}{\alpha}}\right)^m.
\end{align}
}
\end{proposition}
The proof is provided in  Appendix~\ref{app: LT ratio of shot noise process}.

Similar to the result in Proposition~\ref{pro: CCDF ratio of shot noise}, the above Laplace function depends on the ratio $\lambda_2/\lambda_1$ but not the actual values of densities. The case of highly asymmetric shot-noise densities in $R(\lambda_1,\lambda_2)$ corresponds to $\lambda_2/\lambda_1\rightarrow 0$. For this case, the Laplace function has the following simple asymptotic form: 
\begin{align}\nn
\mathcal{L}_{R\left(\lambda_1,\lambda_2\right)}\left(s\right)=\frac{1}{\Gamma(1-\frac{2}{\alpha})}\frac{\lambda_2}{\lambda_1}s^{-\frac{2}{\alpha}} + O\l(\!\l(\!\frac{\lambda_2}{\lambda_1}\!\r)^2\r), \ \frac{\lambda_2}{\lambda_1}\rightarrow 0
\end{align}
that is a linear function of the ratio $\lambda_2/\lambda_1$. 
%Last, it is worth mentioning that the summation in the Laplace function in Proposition xxx always converges though individual terms may not exist for certain values of $\alpha$ (e.g., $\alpha = 2$),  due to that the fact that the Gamma function $\Gamma(x)$ diverges when $x$ is an negative integer. 

% Cache Enabled Distributed Mobile Networks

%%%%%%%%%%%%%%%%%%%%%%%
%%%%%%%      Section 1      %%%%%%%%
%%%%%%%%%%%%%%%%%%%%%%%
\section{Content-Delivery Probability} \label{sec: coverage analysis}

In the preceding  section, we derive the results  concerning  the distribution  of a shot-noise ratio. In this section, they are applied to analyze the content-delivery probability. 

\subsection{Two Forms of Content-Delivery Probability} In this sub-section, the conditional delivery probability defined in Section~\ref{sec: Mathematical Model} is derived in two forms. Note that the expectation of the conditional probability with respect to the content-popularity distribution yields the delivery probability. 

First, the conditional delivery probability is converted into the distribution function of a shot-noise ratio so as to leverage the results derived in the preceding section. One can see from \eqref{Eq: conditional_coverage_def cache one}  that the conditional probability is not exactly the distribution of the shot-noise ratio  but one with a weighted sum of $(N-1)$ shot-noise processes in the denominator. However, the problem can be solved by converting the sum  into a single shot-noise process by exploiting the fundamental characteristic of the shot-noise which is proved to be stable. From Definition~\ref{Def: stable distribution}, stable distribution is invariant to linear operations such as scaling and summation. Consequently, the said weighted sum,  corresponding to the total interference power for the typical user, follows the same stable distribution as a single shot-noise process as shown below. 

\begin{lemma}[Normalized Distributions of Signal and Interference] \label{Lemma: scaling_PPP} \emph{The signal and interference power for the typical user, as given in \eqref{eq: def power with fading}, are distributed as scaled versions of shot-noise processes with unit density: 
\begin{align}
S_0(\mathcal{F}_k) &\sim h_k (a_k\lambda )^{\frac{\alpha}{2}}S(\bar{\Phi}_1),\nn\\
\sum_{n\neq k} I_n &\sim \left(\sum_{n \neq k} {h_n}^{\frac{2}{\alpha}}a_n \lambda\right)^{\frac{\alpha}{2}}S(\bar{\Phi}_2)\nn
\end{align}
where $S(\cdot)$ is a shot noise process without fading as defined in \eqref{Eq:Shot:Def}, and $\bar{\Phi}_1$ and $\bar{\Phi}_2$ are two independent homogeneous PPPs with unit density. 
}
\end{lemma}

The proof is provided in  Appendix~\ref{app: scaling_PPP}. 

Define a normalized shot-noise ratio, denoted as $\bar{R}$, as one generated by $\bar{\Phi}_1$ and $\bar{\Phi}_2$. It follows from Lemma~\ref{Lemma: scaling_PPP} and \eqref{Eq: conditional_coverage_def cache one} that the conditional delivery probability can be written in terms of $\bar{R}$ as:
\begin{equation}
{P}_{d}(\mathcal{F}_k)= \Pr\l({ \bar{R}} \frac{h_k {a_k}^{\frac{\alpha}{2}}}{{g_k}^{\frac{\alpha}{2}}}>\theta_k\r) \label{eq: redefine coverage probability based on scaling ppp}
\end{equation} 
where the set of random variables $\{g_1, g_2, \cdots, g_N\}$ are defined by $g_k=\sum_{n \neq k} a_n{h_n}^{\frac{2}{\alpha}}$. Using the expression, the conditional delivery probability is derived in two forms as follows. 

The first form, called the \emph{expectation form}, follows from substituting the CCDF of a shot-noise ratio as given in Proposition~\ref{pro: CCDF ratio of shot noise} into \eqref{eq: redefine coverage probability based on scaling ppp}, and the result is obtained as follows. 

\begin{theorem} [Content-Delivery Probability: Expectation Form]\label{Th: Coverage Probability expectation form}\emph{Given that the request by the typical user is $D_0 =\mathcal{F}_k$, the conditional delivery  probability is 
 \begin{align}\label{eq: Coverage Probability expectation form}
{P}_{d}(\mathcal{F}_k)=\E_{h_k,g_k}&\left[\arctan\left(\Bigg(1-\frac{2}{1+\left(\frac{h_k}{\theta_k}\right)^{\frac{2}{\alpha}}\frac{a_k}{g_k}}\Bigg)\tan\frac{\pi}{\alpha}\right)\right]\nn\\
&\qquad \qquad \qquad \qquad \qquad  \times \frac{\alpha}{2\pi}+\frac{1}{2}
\end{align}
and the content-delivery probability is ${P}_{d}=\sum_{k=1}^N a_k {P}_{d}(\mathcal{F}_k)$. 
}
\end{theorem}

Let the variable  ${h_k}^{\frac{2}{\alpha}} a_k$ be referred to as the \emph{(fading) weighted popularity} for $\mathcal{F}_k$. Then the ratio $\frac{{h_k}^{\frac{2}{\alpha}} a_k}{g_k}$, called the \emph{weighted popularity ratio},   quantifies  the  weighted popularity of $\mathcal{F}_k$ with respect to that for other files.  One can observe from the result in Theorem~\ref{Th: Coverage Probability expectation form} that the conditional delivery probability for a particular file is a \emph{monotone increasing} function of the corresponding weighted popularity ratio since the function $\arctan$ is  monotone increasing. In addition, the conditional probability reduces with an increasing SIR threshold $\theta_k$.

For the special case of $\alpha$, the expression for the conditional content delivery probability can be simplified as follows. 
\begin{corollary}[Content-Delivery Probability for $\alpha=4$] \label{Co: special case:1}
\emph{ For the path-loss exponent $\alpha=4$, the conditional delivery  probability can be simplified as 
\begin{equation}\label{eq: Coverage special case}
{P}_{d}(\mathcal{F}_k)=1-\frac{2}{\pi}\E_{h_k,g_k}\left[\arctan\left(\frac{g_k}{a_k}\sqrt{\frac{\theta_k}{{h_k}}}\right)\right]. 
\end{equation}
}
\end{corollary}
The proof is provided in Appendix~\ref{app: special case:1}.  

Next, an alternative form of the conditional delivery probability can be obtained by using the Laplace function of a shot-noise ratio in the series form as given in Proposition~\ref{pro: LT ratio of shot noise process}. To this end, it follows from  \eqref{eq: redefine coverage probability based on scaling ppp} and with the fact that $h_k$ is an exponential r.v., the conditional delivery probability can be written as 
\begin{align}
{P}_{d}(\mathcal{F}_k)&=\E_{g_k}\left[e^{-\theta_k\l(\frac{g_k}{a_k}\r)^{\frac{\alpha}{2}}{\bar{R}}^{-1}}\right] \nn\\
&=\E_{g_k}\left[\mathcal{L}_{{\bar{R}}^{-1}}\left(\theta_k\l(\frac{g_k}{a_k}\r)^{\frac{\alpha}{2}}\right)\right]. \nn 
\end{align}
Substituting the result in Proposition~\ref{pro: LT ratio of shot noise process} yields 
\begin{equation}
{P}_{d}(\mathcal{F}_k) =\E_{g_k}\left[\sum_{m=1}^{\infty}\frac{(-1)^{m+1}}{\Gamma(1-m\frac{2}{\alpha})} \l(\frac{a_k}{{\theta_k}^{\frac{2}{\alpha}}g_k}\r)^m \right]. 
\end{equation}
Since ${P}_{d}(\mathcal{F}_k)$ converges, interchanging  the order of summation and expectation in the above expression is allowed  according to Fubini's theorem, leading to the following result. 

\begin{theorem} [Content-Delivery Probability: Series Form]\label{Th: Coverage Probability}\emph{ 
Given that the request by the typical user is $D_0 =\mathcal{F}_k$, the conditional delivery  probability is 
\begin{align} \label{Eq: coverage probability}
{P}_{d}(\mathcal{F}_k)=& \sum_{m=1}^{\infty}\frac{(-1)^{m+1}}{\Gamma(1-m\frac{2}{\alpha})}\E\left[{g_k}^{-m}\right]\l(\frac{a_k}{{\theta_k}^{\frac{2}{\alpha}}}\r)^m 
\end{align}
and the content-delivery probability is ${P}_{d}=\sum_{k=1}^N a_k {P}_{d}(\mathcal{F}_k)$. 
}
\end{theorem}

To relate the results in the current theorem and Theorem~\ref{Th: Coverage Probability expectation form}, consider the case where the SIR threshold $\theta_k$ is large. As a result, the conditional delivery probability in Theorem~\ref{Th: Coverage Probability} can be approximated as {
\begin{align} 
{P}_{d}(\mathcal{F}_k)&\approx  \frac{a_k}{{\theta_k}^{\frac{2}{\alpha}}\Gamma(1-\frac{2}{\alpha})}\E\l[\frac{1}{g_k}\r], \qquad \theta_k \gg 1\nn\\
&\geq \frac{a_k}{{\theta_k}^{\frac{2}{\alpha}}\Gamma(1-\frac{2}{\alpha})}\cdot \frac{1}{\E[g_k]}\label{eq: jensen} \\
&= \frac{a_k}{{\theta_k}^{\frac{2}{\alpha}}\Gamma(1-\frac{2}{\alpha})}\cdot\frac{1}{\sum_{n \neq k} a_n\E\l[{h_n}^{\frac{2}{\alpha}}\r]}  \nn\\
&=  \frac{\sin\l(\pi\frac{2}{\alpha}\r)}{\pi \frac{2}{\alpha} }{\theta_k}^{-\frac{2}{\alpha}}\frac{a_k}{1 - a_k } \label{eq:gamma formula}
\end{align}
where \eqref{eq: jensen} uses Jensen's inequality and \eqref{eq:gamma formula} is obtained using the formula  $\Gamma(1-z)\Gamma(z)=\frac{\pi}{\sin(\pi z)}$ for non-integer $z$.}  It can be observed that ${P}_{d}(\mathcal{F}_k)$ linearly increases with the growing ratio $\frac{a_k}{1 - a_k }$ which is similar to the weighted popularity ratio in the remark on Theorem~\ref{Th: Coverage Probability expectation form}. 

\begin{remark}[Convergence Conditions]\emph{ The convergence of the series in \eqref{Eq: coverage probability} requires that each summation term is finite. Sufficient conditions for convergence can be derived under the following constraint: 
\begin{equation}
\l(\frac{\E\left[{g_k}^{-m}\right]}{\Gamma(1-m\frac{2}{\alpha})}\r)^{\frac{1}{m}}\frac{a_k}{{\theta_k}^{\frac{2}{\alpha}}} < 1, \qquad \forall m. 
\end{equation}
The details are tedious and omitted for brevity. }
\end{remark}

\subsection{Bounding  the  Content-Delivery Probability}

Bounds on the conditional-delivery probabilities whose expressions are simpler than the exact ones are derived in this sub-section. The main idea is to consider the following two modified forms of the SIR in \eqref{Eq:SIR:Modified}, denoted as $\SIR'$ and $\SIR''$, and study their distributions: 
\begin{align}
\SIR'(\mathcal{F}_k) &=\frac{h_k \sum_{X \in \Phi_k }\left| X\right| ^{-\alpha}}{\sum_{z\in \Phi \backslash \Phi_k}h_z \left|Z\right|^{-\alpha}}, \nn\\
 \SIR''(\mathcal{F}_k) &=\frac{h_k \sum_{X \in \Phi_k }\left| X\right| ^{-\alpha}}{h_{k'}\sum_{z\in \Phi \backslash \Phi_k} \left|Z\right|^{-\alpha}}. \label{Eq:SIR:Modified}
\end{align}
They differ from $\SIR$ in \eqref{Eq: conditional_coverage_def cache one}  in the locations of fading coefficients. One can interpret $\SIR'$ and $\SIR''$  as the SIRs corresponding  to two artificial cases: the former  without MAC alignment in the interference and the latter having uniform MAC and an identical transmitted file through out all interference.  Their distributions are related to those of the actual SIR as shown below. 

\begin{lemma}[Bounding the SIR]\label{lemma: Relationships of different forms shot-noise ratio}\emph{The distribution of the SIR in \eqref{Eq: conditional_coverage_def cache one}  can be bounded as 
\begin{align}
\Pr(\SIR(\mathcal{F}_k) \geq \theta_k)&\geq\Pr(\SIR'(\mathcal{F}_k)>\theta_k),\nn\\
\Pr(\SIR(\mathcal{F}_k) \geq \theta_k)&\leq \Pr(\SIR''(\mathcal{F}_k) \geq \theta_k), \quad  \forall\ k 
\end{align}
where $\SIR'$ and $\SIR''$ are defined in \eqref{Eq:SIR:Modified}. 
}
\end{lemma}
The proof is given in  Appendix~\ref{app: Relationships of different forms shot-noise ratio}.
The expressions for the bounds in Lemma~\ref{Lem:Mod:SIR} can be obtained as shown below. 

\begin{lemma}[Modified SIRs] \label{Lem:Mod:SIR}
\emph{The distributions of the modified SIRs in \eqref{Eq:SIR:Modified} are given as 
\begin{align} 
\Pr(\SIR'(\mathcal{F}_k) > x) =&  \E_{h_k}\left[\arctan\left(\Big(1-\frac{2}{1+\eta_k h_k^{\frac{2}{\alpha}}}\Big)\tan\frac{\pi}{\alpha}\right)\right]\nn\\
 &\times\frac{\alpha}{2\pi}+\frac{1}{2}\nn\\
\Pr(\SIR''(\mathcal{F}_k) >x) =& \frac{1}{1+x^{\frac{2}{\alpha}} \l(\frac{1}{a_k}-1\r)}\nn
\end{align} 
where the constant  $\eta_k=\frac{a_k}{(1-a_k)\Gamma(1+\frac{2}{\alpha})x^{\frac{2}{\alpha}}}$.
}
\end{lemma}

The proof is provided in Appendix~\ref{app:Mod:SIR}. Combining Lemmas~\ref{lemma: Relationships of different forms shot-noise ratio} and~\ref{Lem:Mod:SIR} yields the following main result of this sub-section.

\begin{theorem}[Bounding Content-Delivery Probability] \label{Theorem: bounding coverage prob}
\emph{ Using the result from Lemma~\ref{lemma: Relationships of different forms shot-noise ratio}, conditional delivery probability is bounded by 
\begin{align} \label{eq: bounds for general case}
{P}_{d}(\mathcal{F}_k)&\geq\frac{1}{2}+\frac{\alpha}{2\pi}\E_{h_k}\left[\arctan\left(\Big(1-\frac{2}{1+\eta_k h_k^{\frac{2}{\alpha}}}\Big)\tan\frac{\pi}{\alpha}\right)\right], \nn\\
{P}_{d}(\mathcal{F}_k)&\leq \frac{1}{1+\theta_k^{\frac{2}{\alpha}} \left(\frac{1}{a_k}-1\right)}
\end{align} 
where the constant $\eta_k=\frac{a_k}{(1-a_k)\Gamma(1+\frac{2}{\alpha}){\theta_k}^{\frac{2}{\alpha}}}$.
}
\end{theorem}

\begin{remark} [Skewed Popularity Distribution]\label{Re:Skew}\emph{
For a sanity check, one can see that for a singularly popular file with $a_k \rightarrow 1$, both upper and lower bounds on the conditional delivery probability grow with $a_k$ and converge to $1$. Correspondingly, for a highly skewed popularity distribution, CAMAC effectively turns the network into one mostly broadcasting a single popular file, for which interference is negligible.} 
\end{remark}

\begin{remark} [Effect of Path-Loss Exponent]\label{Re:pathloss}\emph{
In radio access networks without multi-cell cooperation, increasing the path-loss exponent reduces the number of significant interferers and thereby improving link reliability measured by e.g., outage probability (see e.g., \cite{andrews2011tractable}). Consequently, the path-loss exponent has a significant effect on network performance. In contrast, for the current content-delivery network with CAMAC,  there exists multiple serving helpers and interferers for each user.  Increasing the path-loss exponent reduces simultaneously the numbers of helpers and interferers. As the result, the effect of path-loss exponent on network performance is not significant. As an example, given $\theta = 2$, the relevant factor $\theta^{\frac{2}{\alpha}}$ in the upper bound on the content-delivery probability in \eqref{eq: bounds for general case} is 1.59 for $\alpha = 3$ and a similar value of 1.41 for $\alpha = 4$. The fact is confirmed by simulation in the sequel.} 
\end{remark}

Last, following  Corollary~\ref{Co: special case:2}, the bounds in Theorem~\ref{Theorem: bounding coverage prob} can be further simplified for the special case of $\alpha = 4$ as follows. 
\begin{corollary}[Bounding Content-Delivery Probability for $\alpha=4$] \label{Co: special case:2}
\emph{ For the path-loss exponent $\alpha=4$, the conditional delivery  probability can be bounded as follows. 
\begin{itemize}
\item Upper bound: 
\begin{equation}
{P}_{d}(\mathcal{F}_k) \leq\frac{1}{1+\sqrt{\theta_k} \l(\frac{1}{a_k}-1\r)}. 
\end{equation}
\item Lower  bound A: 
\begin{equation}
{P}_{d}(\mathcal{F}_k) \geq \frac{1}{\sqrt{\pi}}e^{\zeta_k} \Gamma\left(\frac{1}{2}, \zeta_k\right)\nn
\end{equation}
where the constant $\zeta_k = \frac{\pi\theta_k}{4}\left(\frac{1-a_k}{a_k}\right)^2$. 
\item Lower  bound B: 
\begin{equation}
{P}_{d}(\mathcal{F}_k) \geq 1-\frac{2}{\pi}\arctan\left(\frac{\pi\sqrt{\theta_k}}{2}\Big(\frac{1}{a_k}-1\Big)\right). \label{eq: concise LB}\nn
\end{equation}
\end{itemize}
}
\end{corollary}
The proof is provided in Appendix~\ref{app: special case:2}.

\section{Spatial Alignment Gain}\label{sec: discussion}
In this section, we attempt to quantify the network-performance gain of CAMAC with respect to the conventional case without using the algorithm. The (MAC) \emph{spatial alignment gain}, denoted as $G_{\text{align}}$, is defined as the ratio of content-delivery probabilities  between   the current and the mentioned conventional case. 

Consider the case with CAMAC. Simulation results (see Fig.~\ref{Fig: Verify}) shows that the upper bound on the conditional delivery probability in Theorem~\ref{Theorem: bounding coverage prob} is sufficiently tight. Thus, we approximate the probability by the bound so as to simplify the analysis: 
\begin{equation} \label{eq: approximate CAMAC delivery probability}
P_d(\mathcal{F}_k) \approx \frac{1}{1+\theta_k^{\frac{2}{\alpha}} \left(\frac{1}{a_k}-1\right)}. 
\end{equation}

Next, for the conventional case without CAMAC, the conditional delivery probability, denoted as $\tilde{P}_d(\mathcal{F}_k)$, is given as 
\begin{align}\label{eq: non-synchronization def converage }
\tilde{P}_d(\mathcal{F}_k) =  \Pr\left(\frac{h r_k^{-\alpha}}{\sum _{X \in \Phi \backslash \{ T_k\} }h_X \left| X \right|^{-\alpha}}> \theta_k \right)
\end{align}
where $r_k$ is the distance between the  typical user to the nearest  helper $T_k$ having $\mathcal{F}_k$. The procedure for deriving  a closed-form expression for the above $\tilde{P}_d(\mathcal{F}_k)$ is standard  in the literature of stochastic geometry, involving essentially deriving the Laplace function of a shot-noise process using Campbell's Theorem \cite{martin}. It is straightforward to modify the existing results e.g., \cite[Lemma 3]{wen2017cache} to obtain the following lemma. 

\begin{lemma}[Content-Delivery Probability without CAMAC]\label{lemma: Coverage probability for non-synchronization network}\emph{For the conventional case without CAMAC, the conditional content-delivery probability is given as 
\begin{equation} \label{eq: Prob without alignment}
\tilde{P}_{d}(\mathcal{F}_k)=\frac{1}{1+ \mu \left(\theta_k, \alpha \right)+\theta_k^{\frac{2}{\alpha}}\frac{2\pi}{\alpha}\csc\left(\frac{2\pi}{\alpha}\right)\l(\frac{1}{a_k}-1\r)} 
\end{equation}
where the function $\mu$ is defined as 
\begin{equation}
\mu\left( \theta_k, \alpha \right)=\int_{1}^{\infty}\frac{1}{1+\theta_k^{-1}x^{\frac{\alpha}{2}}}dx.
\end{equation}
}
\end{lemma}
Since  the scaling factor $\frac{2\pi}{\alpha}\csc\left(\frac{2\pi}{\alpha}\right)\geq 1$, it follows from the above result that 
\begin{align} \label{eq: upper bound coverage probability without alignment}
\tilde{P}_{d}(\mathcal{F}_k)\leq \frac{1}{1+ \mu \left(\theta_k, \alpha \right)+\theta_k^{\frac{2}{\alpha}}\l(\frac{1}{a_k}-1\r)}.
\end{align}

It is ready to analyze the spatial-alignment gain, $G_{\text{align}}$, defined earlier. Consider the scenario of a highly skewed popularity distribution, corresponding to $\gamma\gg 1$ and $\gamma \approx  0$, respectively. As a result, $\mathcal{F}_1$ is dominant  with $a_1 \approx 1$ or equally popular as others with $a_1 \approx \frac{1}{N}$. Then the content-delivery probabilities can be approximated by the conditional probabilities for $D_0 = \mathcal{F}_1$: $P_{d}\approx P_{d}(\mathcal{F}_1)$ and $\tilde{P}_{d}\approx \tilde{P}_{d}(\mathcal{F}_1)$. Using the results in \eqref{eq: approximate CAMAC delivery probability} and  Lemma~\ref{lemma: Coverage probability for non-synchronization network}, the spatial-alignment gain is approximated as
\begin{align} \label{eq: alignment gain}
\boxed{G_{\text{align}} \approx 1+\frac{\mu\left(\theta_1, \alpha\right)}{1+{\theta_1}^{\frac{2}{\alpha}}\l(\frac{1}{a_1}-1\r)}}\ . 
\end{align}
One can see that the $G_{\text{align}}$ grows with the popularity measure $a_1$. As $a_1$ approaches $1$, $G_{\text{align}}$ converges to the limit of $1 + \mu\left(\theta_k, \alpha\right)$. This is the inverse of the limit of $\tilde{P}_{d}(\mathcal{F}_k)$ while that of $P_{d}(\mathcal{F}_k)$ is one. The implication is twofold. First, the network-performance gain due to CAMAC grows with the  skewness of the  popularity distribution.  Second, given the highly skewed  distribution,   a network with CAMAC operates in the noise-limited regime while a conventional network is interference limited.

\section{Simulation Results}\label{sec: simulation results}
The default  simulation settings are  as follows unless specified otherwise.  The content helpers are Poisson distributed with density $\lambda=0.1$ and all SIR threshold are given as $\theta_k=5$ for simplicity.  The path-loss exponent is set as $\alpha=3\ \mathrm{and} \ 4$ corresponding to the general case and special case.  For the popularity distribution, we adopt the widely used Zipf distribution with the content-popularity skewness  $\gamma \in [0,3]$: $ a_n=\frac{n^{-\gamma}}{\sum_{m=1}^{N}  m^{-\gamma}}$ for all $n$.  Last, the benchmarking case without CAMAC has the content-delivery probability given in \eqref{eq: Prob without alignment}. 

\begin{figure}[t!]
\begin{center}
\subfigure[Path-loss exponent $\alpha = 3$]{\includegraphics[width=7cm]{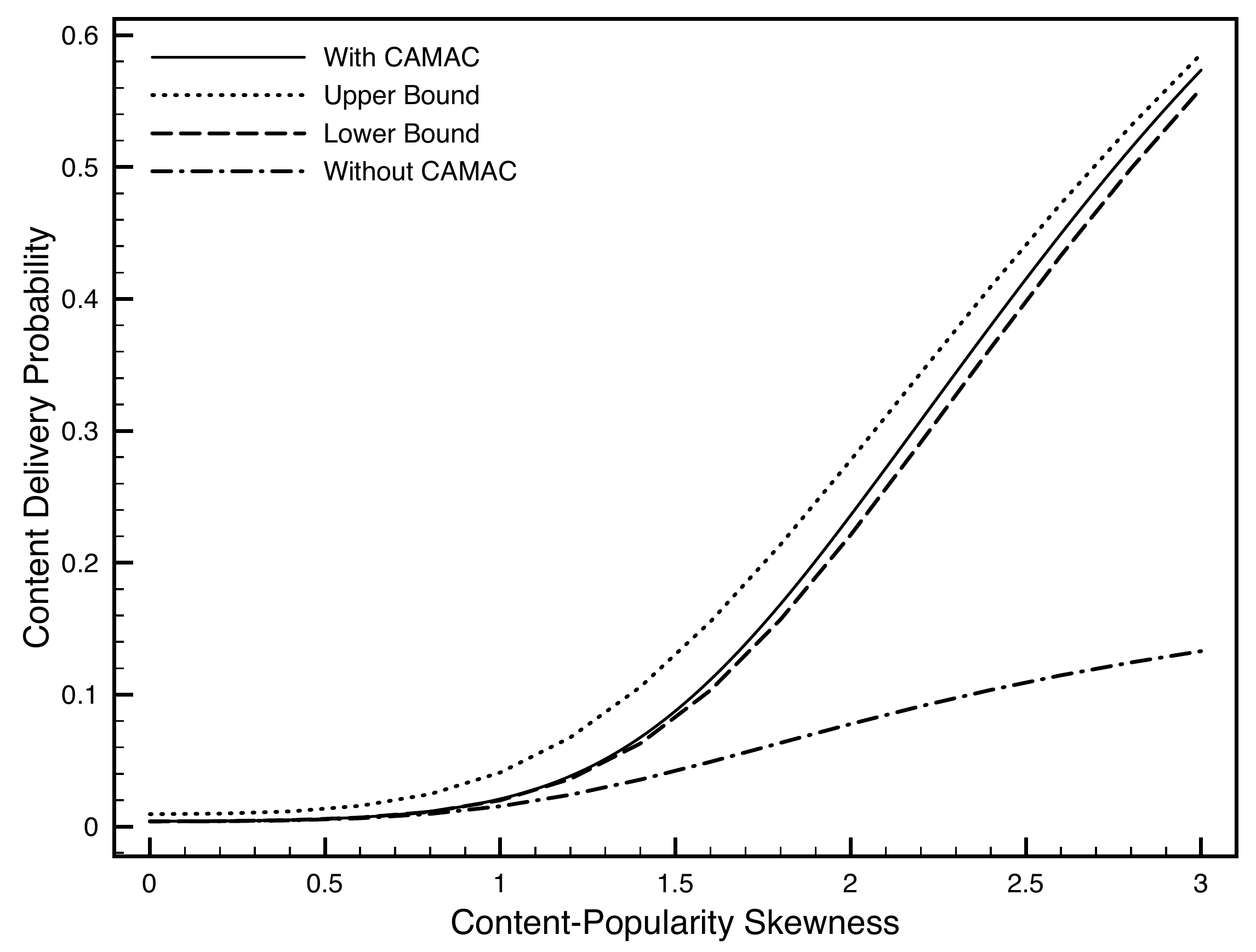}}\\
\subfigure[Path-loss exponent $\alpha = 4$]{\includegraphics[width=7cm]{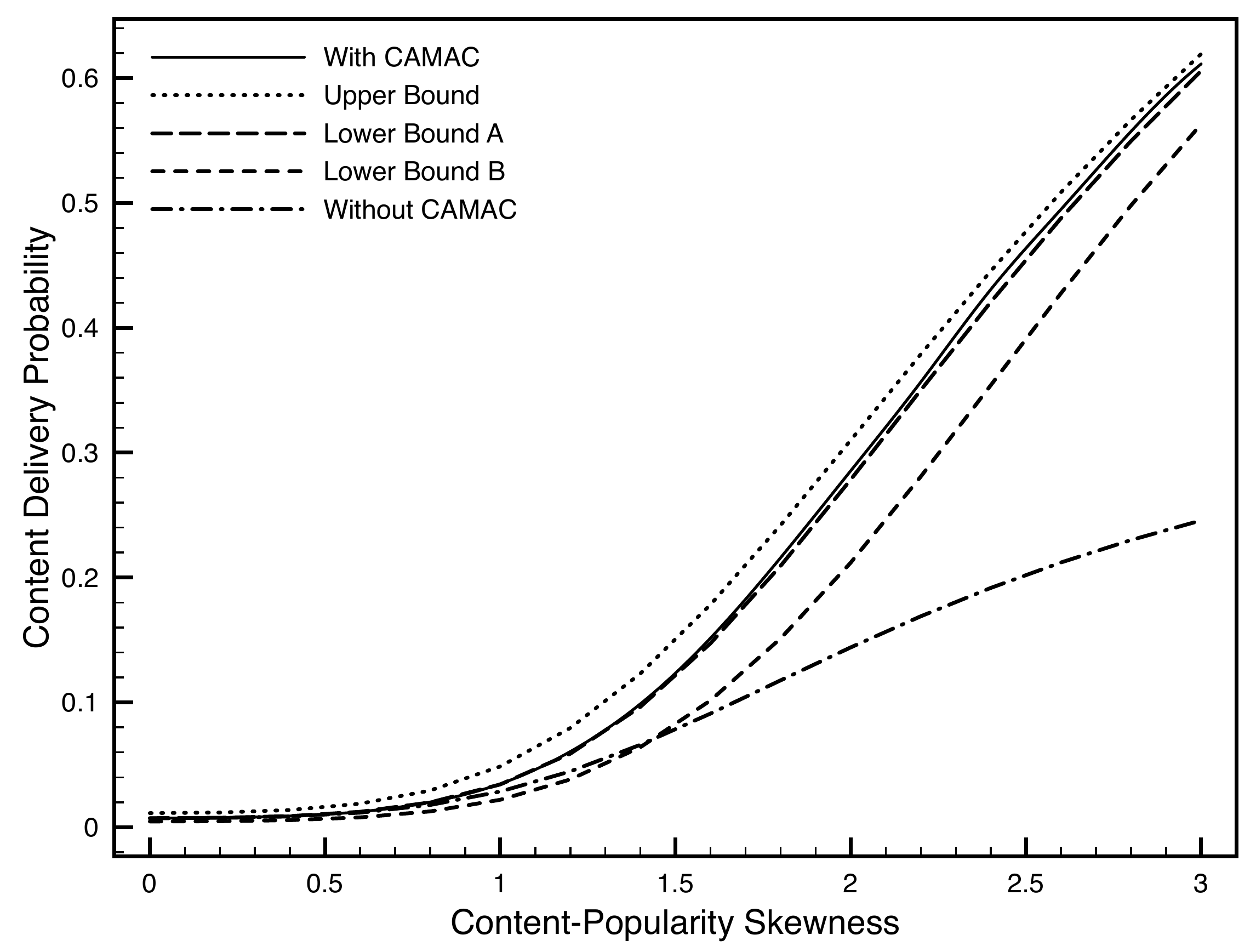}}
\caption{Comparisons of content delivery probability between the cases with and without CAMAC and evaluation of the derived bounds on the probability. The  number of files is $N=50$.}
\label{Fig: Verify}
\end{center}
\end{figure}

Fig.~\ref{Fig: Verify} displays the curves of convent-delivery probability versus the content-popularity skewness for both the cases with and without CAMAC. Moreover, the derived bounds on the probability in Theorem~\ref{Theorem: bounding coverage prob} are also plotted for evaluating their tightness. One can observe that as the skewness increases,  the delivery probability grows faster in the case with CAMAC than the conventional case.  In particular, for skewness of $3$, the spatial-alignment gain reaches about $6$ and $3$ for $\alpha$ equal to $3$ and $4$, respectively. The smaller gain for a larger path-loss exponent is due to that more severe propagation attention suppresses interference in the conventional case and thereby helps content delivery. However,  the path-loss exponent has  a negligible effect when CAMAC is used, for which the network is not interference limited. Next, it is observed that the upper bound and especially the lower bound  derived in Theorem~\ref{Theorem: bounding coverage prob} are tight. For the case of $\alpha = 4$, there exist lower bounds A and B (see Corollary~\ref{Co: special case:2}). The former is observed to be tighter than the latter that, however, is also  capable of tracing the variation of delivery probability.  {Last, it is worth mentioning that, the content-delivery probability for CAMAC  can be further improved by integrating the scheme with advanced communication techniques such as multi-antenna transmissions (see e.g., \cite{huang2013analytical}) and  successful interference cancellation \cite{weber2007transmission}.
}

The number of files in the content-data base can affect content delivery probability as revealed in Fig.~\ref{Fig: EffectN}.  {The analytical results in Theorems 1 and 2 are found to be identical with the simulation results shown in Fig.~\ref{Fig: EffectN}.} When the content-popularity skewness is small, different files have comparable popularity. In this case, one can observe Fig.~\ref{Fig: EffectN} that a smaller number of files amplifies the gain of CAMAC in terms of delivery probability due to larger set of signals aligned and combined by CAMAC and the corresponding reduction on interference.  Nevertheless, the differentiation in delivery probability for varying the database size diminishes as the skewness growth. For this case, the popularity is concentrated in one or a few files regardless of the database size.

\begin{figure}[t!]
\begin{center}
{\includegraphics[width=7cm]{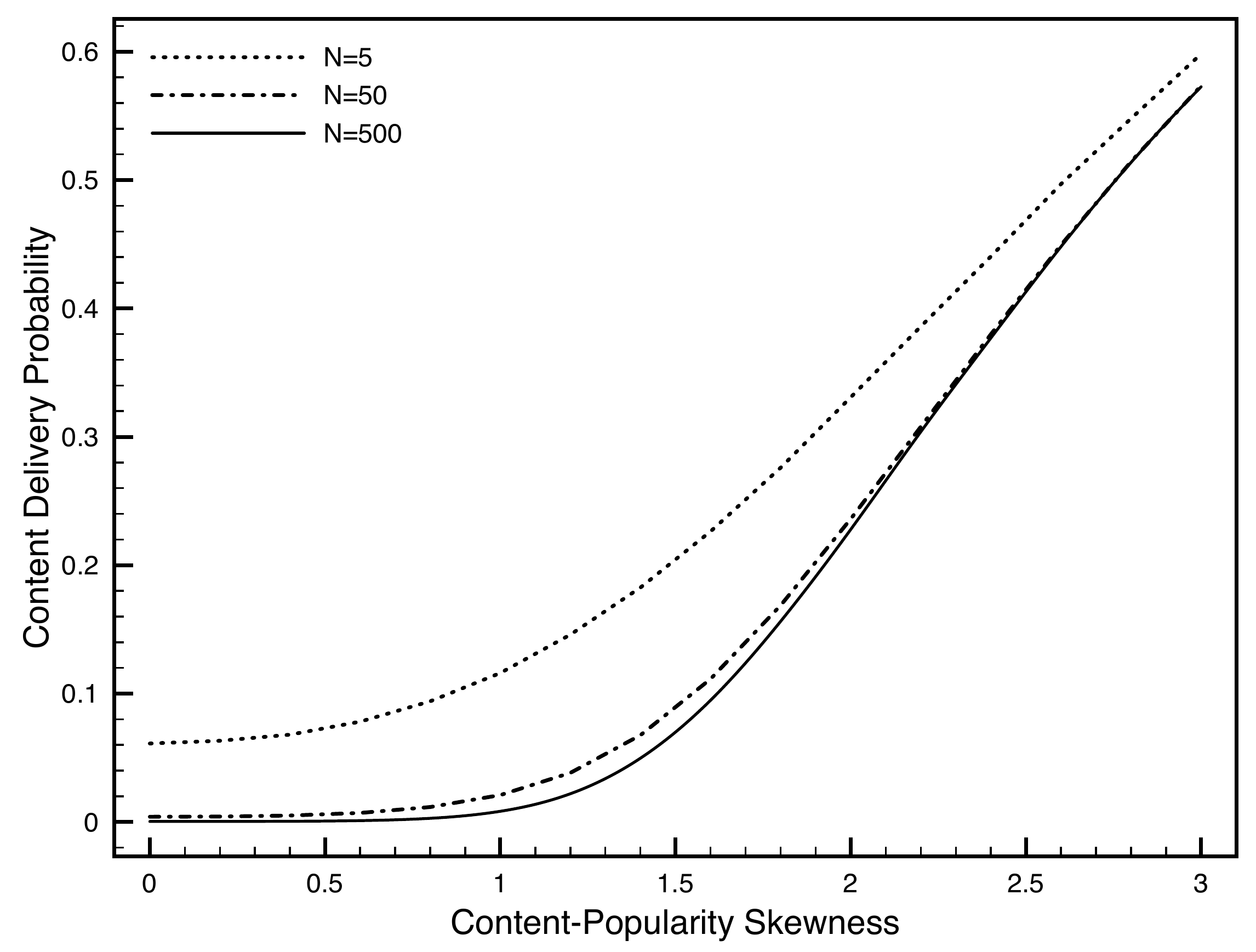}}\\
\caption{The effect of the number of files on the convent-delivery probability.}
\label{Fig: EffectN}
\end{center}
\end{figure}

The approximation of the spatial alignment gain as derived in \eqref{eq: alignment gain} is shown to be sufficiently  accurate by comparing with the exact simulation results in Fig.~\ref{Fig: comb gain limited caching}.   It is observed that the approximation is accurate throughout the considered range of skewness for different database sizes, even though deriving the result assumes  the skewness being either high or low. 

\begin{figure}[t!]
\begin{center}
{\includegraphics[width=7cm]{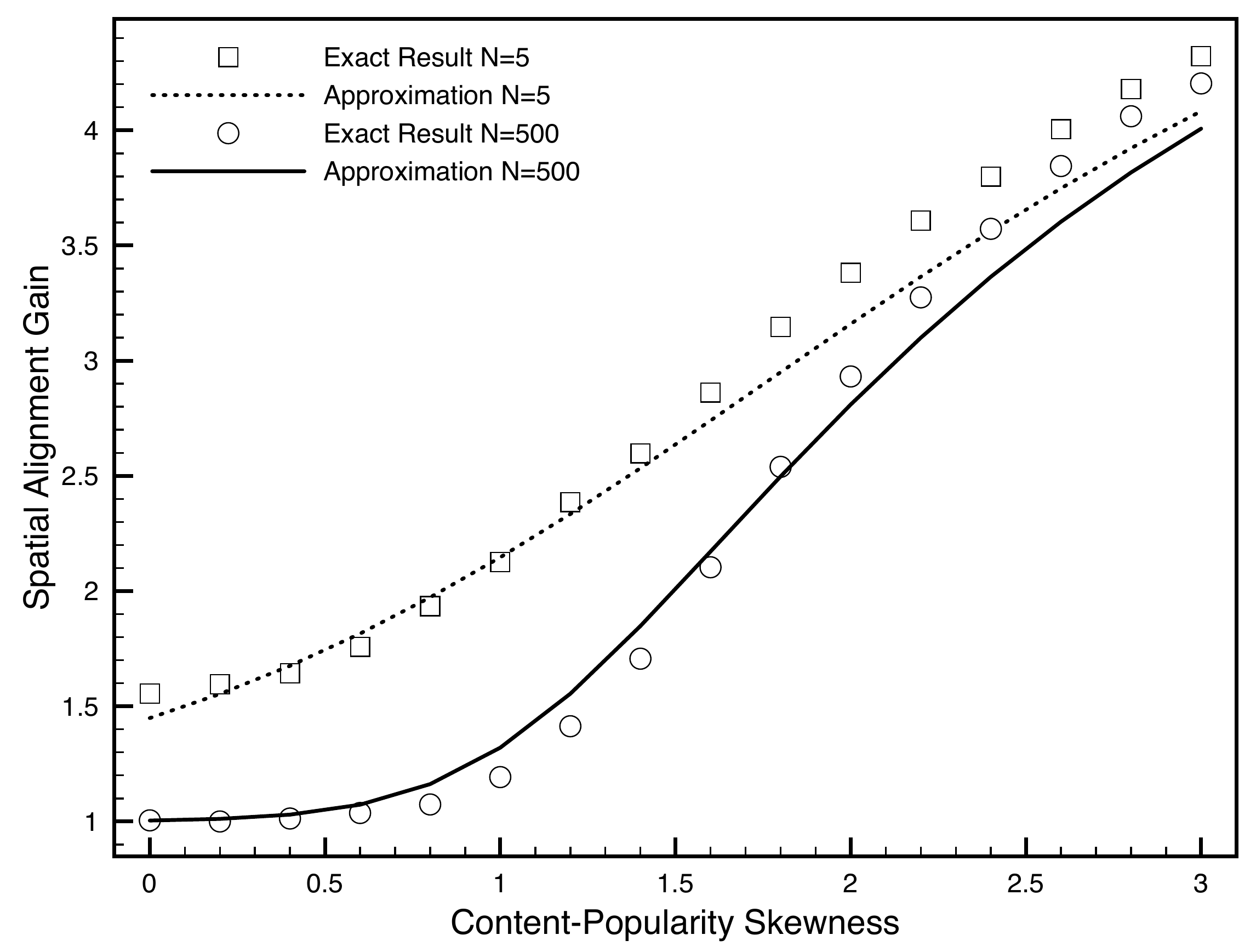}}\\
\caption{The spatial alignment gain and its approximation versus the content-popularity skewness. The number of files is $N = 5$ or $500$. }
\label{Fig: comb gain limited caching}
\end{center}
\end{figure}

\section{Concluding Remarks}\label{sec: Concluding Remarks}
In this work, we have proposed the idea of content adaptive modulation and coding (CAMAC) and quantify how much performance gain the idea can bring to a content delivery network. Through this work, we have shown that the spatial signal correlation in such a network can be exploited in simple ways to substantially enhance the reliability of content delivery.  {The performance of content-delivery networks with CAMAC can be improved by optimizing random caching strategies in the same vein as \cite{chen2017probabilistic,ji2016wireless,wen2017cache,chen2017cooperative,wen2017random}. The performance gain is expected to be significant especially for the scenario of heterogeneous helpers with e.g.,  different storage capacities and transmission powers. The design will involve the interplay between optimization theory and the approach of shot-noise ratio developed in this work. } As the current work represents an initial study of CAMAC, the implementation of the design requires further research to address various practical issues.  In particular, how to cope  with the frequency selectivity in the effective channels induced by CAMAC. Moreover, how to optimize the content-MAC mapping based on different quality-of-service requirements for different types of content. 

In the process of network-performance analysis using a stochastic-geometry model, we have solved an open problem of finding a tractable approach for analyzing the ratio of shot-noise processes. In general, the approach and the derived results can find other applications in studying general networks with cooperative transmissions beyond the current content delivery networks with CAMAC. 

\appendix

%   subsection  stable distribution
\subsection{Some Useful Properties of Stable Distribution}\label{app: stable distribution}
The characteristic function for the class of stable distribution defined in Definition~\ref{Def: stable distribution} can be expressed in two different forms as shown below \cite{zolotarev-stable-distribution}. 

 \begin{lemma}[Characteristic Function for Stable Distribution \cite{zolotarev-stable-distribution}] \label{lemma: CF stable distribution}\emph{
The logarithm of the characteristic function for a r.v. $X$ belonging to the class of stable distribution can be written in different forms: 
\begin{itemize}
\item  \textbf{Form A }
\begin{align}\label{eq: def form A}
\log \E\left[e^{jtX}\right]=\left\{\begin{aligned}
jt\mu_{A}\gamma_{A}-\mu_{A}\left|t\right|^{\delta}&\Big(1-j \beta_A \  \mathrm{sgn}(t)\\
&\times\tan\frac{\pi \delta}{2}\Big), \ \delta\neq 1\\ 
jt\mu_{A}\gamma_{A}-\mu_{A}\left|t\right|&\Big(1+j \beta_A\frac{2}{\pi}   \ \mathrm{sgn}(t)\\
&\times \log|t|\Big), \  \delta=1
\end{aligned}\right.
\end{align}
where the real parameters satisfy  $\delta\in [0, 2]$, $\beta_A \in [-1, 1]$, $\gamma_A\in \mathds{R}$, $\mu_A>0$,  and the sign function $\mathrm{sgn}(t)$ is defined as: 
\begin{align}\label{eq: sgn}
\mathrm{sgn}(t)=\left\{\begin{aligned}
&1, &t>0\\ 
&0, &t=0\\ 
&-1, &t<0.
\end{aligned}\right.
\end{align}
%%%%  Definition Form B   %%%%
\item  \textbf{Form B }
\begin{align}\label{eq: def form B}
\log \E\left[e^{jtX}\right]=\left\{\begin{aligned}
jt\mu_{B}\gamma_{B}-\mu_{B}\left|t\right|^{\delta}\exp&\Big(-j \beta_B\frac{\pi}{2} \mathrm{sgn} (t) \\
&\times  K(\delta)\Big), \ \delta\neq 1\\ 
jt\mu_{B}\gamma_{B}-\mu_{B}\left|t\right|\Big(\frac{\pi}{2}&+j \beta_B\mathrm{sgn}(t)  \\
&\times \log|t|\Big), \ \delta=1
\end{aligned}\right.
\end{align}
where the parameters satisfy the same constraints as for Form A and the function  $K(\delta)=\delta-1+\mathrm{sgn}(1-\delta)$. 
\end{itemize}}
\end{lemma}

For ease of notation, the class of stable distribution in Forms A and B are denoted as  $S_A(\delta, \beta_A, \gamma_A, \mu_A)$  and $S_B(\delta, \beta_B, \gamma_B, \mu_B)$, respectively. 

As shown in Lemma~\ref{lemma: CF stable distribution}, a stable distribution is characterised by four parameters. The parameters $\gamma$ and $\mu$ are  the location (or shift) and scale parameters, respectively.   On the other hand,  $\delta$ and $\beta$ essentially specify  the shape of PDF. To be specific, $\delta$ is called the  index of stability or characteristic exponent that  is a measure of concentration;  $\beta$ is called the skewness parameter that  is a measure of asymmetry. By comparing  Forms A and B in Lemma~\ref{lemma: CF stable distribution}, the relations between their  parameters are given as: 
\begin{itemize}
\item For $\delta$=1, 
\begin{equation}
\beta_A=\beta_B, \quad \gamma_A=\frac{2\gamma_B}{\pi}, \quad \mu_A=\frac{\pi \mu_B}{2}.  \label{eq: relation equal to 1}
\end{equation}
\item For $\delta\neq1$, 
\end{itemize}
\begin{equation} \label{eq: Relation Form A and B delta neq 1}
\left\{
\begin{aligned}
\beta_A&=\cot \left(\frac{\pi\delta}{2}\right)\tan\l(\frac{\pi\beta_B K(\delta)}{2}\r),\\
\gamma_A&=\gamma_B \l(\cos\l(\frac{\pi \beta_B K(\delta)}{2}\r)\r)^{-1},   \\
\mu_A&=\mu_B\cos\l(\frac{\pi \beta_B K(\delta)}{2}\r).
\end{aligned}
\right.
\end{equation}

For several special cases, the distribution function of a stable r.v. has simple forms as shown below. 

 \begin{lemma}[Properties of Stable Distribution \cite{zolotarev-stable-distribution}]\label{lemma: Special cases of stable distribution}\emph{The distribution function of a stable r.v., $X$, is given for several special cases as follows. 
\begin{itemize}
\item[1)]  If $X\sim S_B(\delta, \beta_B, 0, 1)$ and $\delta \neq 1$, then 
\begin {align} \label{eq:crossing 0 prob}
\Pr(X<0)=\frac{1}{2}\left(1-\beta_B \frac{K(\delta)}{ \delta}\right).
\end{align}
\item[2)] If $X\sim S_B(\delta, 1, 0, 1)$ and $\delta< 1$, then $\Pr(X<x)=0$ for all $x<0$.  
\item[3)] If $X\sim S_B(\delta, -1, 0, 1)$ and $\delta< 1$, then $\Pr(X<x)=1$ for all $x>0$.
\end{itemize}}  
\end{lemma}

\subsection{Proof of Lemma~\ref{pro: CF marked point process}}\label{app: CF marked point process}
The characteristic function can be written in the product form
\begin{align}
\E\l[e^{jtM(x; \lambda_1, \lambda_2)}\r]=\E&\left[e^{j t \sum_{X \in \Phi_1}|X|^{-\alpha}}\right] \nn\\
&\times\E\left[e^{j t\left(-x \sum_{Z\in \Phi_2}|Z|^{-\alpha}\right)}\right]. \label{eq: product form CF DiffShot}
\end{align}
The first term in \eqref{eq: product form CF DiffShot} is the  characteristic function of a shot noise process that is well known and given as (see e.g., \cite{martin}) 
\begin{align}
\E\left[e^{j t \sum_{X \in \Phi_1}|X|^{-\alpha}}\right]&=\exp\left(j \lambda_1 \pi t \int_{0}^{\infty} x^{-\frac{2}{\alpha}}e^{j t x} dx\right) \nn\\
&=e^{-\lambda_1 \pi \Gamma\left(1-\frac{2}{\alpha}\right)\left(-jt\right)^{\frac{2}{\alpha}}}. \label{Eq: CF single shot noise} 
\end{align}
Similarly, the second term in \eqref{eq: product form CF DiffShot} is obtained as 
\begin{align}
\E\left[e^{j t\left(-x \sum_{Z\in \Phi_2}|Z|^{-\alpha}\right)}\right]
&=e^{-\lambda_2  \pi \Gamma\left(1-\frac{2}{\alpha}\right)\left(jtx\right)^{\frac{2}{\alpha}}}. \label{Eq: CF single negative shot noise} 
\end{align}
Substituting \eqref{Eq: CF single shot noise} and \eqref{Eq: CF single negative shot noise} into \eqref{eq: product form CF DiffShot} leads to  the following expression for the characteristic  function: 
\begin{equation} 
G\left(t; \lambda_1, \lambda_2\right)=e^{-\lambda_1 \pi \Gamma\left(1-\frac{2}{\alpha}\right)\left(-jt\right)^{\frac{2}{\alpha}}}\times e^{-\lambda_2  \pi \Gamma\left(1-\frac{2}{\alpha}\right)\left(jtx\right)^{\frac{2}{\alpha}}}.  \label{eq: mid step CF differential shot noise}
\end{equation}
Using the elementary identities $j=e^{j\frac{\pi}{2}}$ and $e^{j\frac{\pi}{\alpha}}=\cos\frac{\pi}{\alpha}+j \sin\frac{\pi}{\alpha}$, the expression can be rewritten in the desired form in the lemma statement. 

\subsection{Proof of Lemma~\ref{lemma: transfer differential shot noise}} \label{app: transfer differential shot noise}
Using Proposition~\ref{pro: marked PP is stable distribution} and Lemma~\ref{lemma: CF stable distribution} and parametric relations in  \eqref{eq: Relation Form A and B delta neq 1} in Appendix~\ref{app: stable distribution}, the characteristic function of $M(x; \lambda_1, \lambda_2)$ is given as 
\begin{equation}
\E\l[e^{j t M(x; \lambda_1, \lambda_2)}\r] = e^{-\mu_{B}\left|t\right|^{\frac{2}{\alpha}}\exp\left(-j \beta_B\frac{\pi}{\alpha}  \mathrm{sgn} (t) \right)}. 
\end{equation}
Assume that the equivalence in \eqref{Eq:Sim:2} holds. Then the characteristic function of $\widetilde{M}(x; \lambda_1, \lambda_2)$ with the variable $t'$  follows from the above equation by substituting
$t = t' \times {\mu_B}^{-\alpha/2}$. 
As a result, 
\begin{equation}
\E\l[e^{j t' \widetilde{M}(x; \lambda_1, \lambda_2)}\r] = e^{-\left|t'\right|^{\frac{2}{\alpha}}\exp\left(-j \beta_B\frac{\pi}{\alpha}  \mathrm{sgn} (t) \right)}. 
\end{equation}
Comparing the expression with the characteristic function of  Form-B stable distribution in Lemma~\ref{lemma: CF stable distribution} in Appendix~\ref{app: stable distribution} gives the equality in \eqref{Eq:Sim:1}, confirming that $\widetilde{M}(x; \lambda_1, \lambda_2)$ is the normalized differential shot-noise and that the assumed result in \eqref{Eq:Sim:2} holds. This completes the proof.

\subsection{Proof of Proposition~\ref{pro: LT ratio of shot noise process}}\label{app: LT ratio of shot noise process}
To get the result of \eqref{eq: mid LT shot-noise ratio}, we obtain $\E\left[e^{-t\left(\sum_{Z\in\Phi_2}\left|Z\right|^{-\alpha}\right)^{-\frac{2}{\alpha}}}\right]$ first by using the series form of PDF of $S(\Phi)$ provided in \eqref{pdf: shot noise}.  Accordingly, 
\begin{align}\label{eq: E of shot noise to arbitrary power}
\E\left[e^{-t \left(\sum_{Z\in \Phi_2}\left|Z\right|^{-\alpha}\right)^{-\frac{2}{\alpha}}}\right]&\nn\\
 =\int_{x>0}e^{-tx^{-\frac{2}{\alpha}}}\frac{1}{\pi x }\sum_{m=1}^{\infty}&\frac{(-1)^{m+1}\Gamma(1+m\frac{2}{\alpha})\sin \pi m \frac{2}{\alpha}}{m!} \nn\\
 &\times \left(\frac{\lambda \pi \Gamma\left(1-\frac{2}{\alpha}\right)}{x^{\frac{2}{\alpha}}}\right)^m dx .
\end{align}
The integral in \eqref{eq: E of shot noise to arbitrary power} is convergent and thus according to Fubini's theorem, we can interchange the order of integral and summation, which yields 
\begin{align}\label{eq: E of shot noise to arbitrary power final}
&\E\left[e^{-t \left(\sum_{Z\in \Phi_2}\left|Z\right|^{-\alpha}\right)^{-\frac{2}{\alpha}}}\right]\nn\\
 &=\sum_{m=1}^{\infty}\frac{(-1)^{m+1}\Gamma(1+m\frac{2}{\alpha})\sin \pi m\frac{2}{\alpha}}{\pi m!}\left(\lambda \pi \Gamma\left(1-\frac{2}{\alpha}\right)\right)^m \nn\\
&\qquad \qquad \qquad \qquad \qquad \qquad \times \int_{0}^{\infty}e^{-tx^{-\frac{2}{\alpha}}}x^{-\frac{2}{\alpha}m-1} dx \nn\\
&\overset{\left( a \right)}{=}\sum_{i=1}^{\infty}\frac{\left(-1\right)^{m+1}}{\Gamma(1-m\frac{2}{\alpha})}\left(\frac{\lambda \pi \Gamma(1-\frac{2}{\alpha})}{t}\right)^m
\end{align}
where $(a)$ is derived from $\Gamma(1-z) \Gamma(z)=\frac{\pi}{\sin(\pi z)}$ for all non-integer $z$.  Substituting $t=\lambda_1 \pi \Gamma(1-\frac{2}{\alpha})s^{\frac{2}{\alpha}}$, $\lambda=\lambda_2$ and combing \eqref{eq: mid LT shot-noise ratio} and \eqref{eq: E of shot noise to arbitrary power final}, we can obtain the desired result.

\subsection{Proof of Lemma~\ref{Lemma: scaling_PPP}} \label{app: scaling_PPP}

First, we obtain the invariance property of shot noise with respect to linear operations. Consider two independent homogeneous PPPs $\Phi_1$ and $\Phi_2$ with density $\lambda_1$ and $\lambda_2$ separately.  Given two positive constants $a$ and $b$, the weighted sum of two shot-noise processes $a S(\Phi_1) + b S(\Phi_2)$ has a Laplace function obtained as 
\begin{align*}
&\E\left[e^{-s \left(a\sum_{X\in \Phi_1\left(\lambda_1\right)}\left|X\right|^{-\alpha}+b\sum_{Z\in \Phi_2\left(\lambda_2\right)}\left|Z\right|^{-\alpha}\right)}\right] \nn\\
&\overset{\left( a \right)}{=}\E \left[ \prod_{X\in \Phi_1}e^{-s a\left|X\right|^{-\alpha}}\right]\E \left[ \prod_{Z\in \Phi_2}e^{-s b\left|Z\right|^{-\alpha}}\right] \\
&=e^{- \pi \Gamma\left(1-\frac{2}{\alpha}\right)\left(\left(\lambda_1 a^{\frac{2}{\alpha}}+\lambda_2 b^{\frac{2}{\alpha}}\right)^{\frac{\alpha}{2}} s\right)^{\frac{2}{\alpha}}}
\end{align*}
where $(a)$ applies  Campbell's theorem \cite{kingman1993poisson}. It can be observed that the result is equivalent to the Laplace function of  $\left(a^{\frac{2}{\alpha}}\lambda_1+b^{\frac{2}{\alpha}}\lambda_2\right)^{\frac{\alpha}{2}}S(\bar{\Phi})$ with $\bar{\Phi}$ being a PPP with unit density. Then the desired results follow from applying this property to  the signal-and-interference expressions.

\subsection{Proof of Corollary~\ref{Co: special case:1}} \label{app: special case:1}
When $\alpha=4$, the conditional content delivery probability in Theorem~\ref{Th: Coverage Probability expectation form} is simplified as  
\begin{align*}
{P}_{d}(\mathcal{F}_k)=\frac{1}{2}+\frac{2}{\pi}\E_{h_k,g_k}\left[\arctan\left(\frac{1-g_k\sqrt{\frac{\theta_k}{h_k}}}{1+g_k\sqrt{\frac{\theta_k}{h_k}}}\right)\right] \nn\\
= \frac{1}{2}+\frac{2}{\pi}\E_{h_k,g_k}\left[\arctan\!\!\left(\tan\!\!\left(\frac{\pi}{4}-\arctan\left(g_k\sqrt{\frac{\theta_k}{h_k}}\right)\!\!\right)\!\!\right)\right]
\end{align*}
which leads to the result in \eqref{eq: Coverage special case}.

\subsection{Proof of Lemma~\ref{lemma: Relationships of different forms shot-noise ratio}}\label{app: Relationships of different forms shot-noise ratio}
Conditioned of the process $\Phi$, the delivery probability based on $\SIR'$ can be expressed as
\begin{align}
\Pr(\SIR '>\theta_k \mid \Phi)&=\E_{h_z}\!\!\!\left[\exp\left(-\frac{\theta_k\sum_{z\in \Phi \backslash \Phi_k}h_z \left|Z\right|^{-\alpha}}{\sum_{X \in \Phi_k }\left| X\right| ^{-\alpha}}\right) \!\! \Big | \Phi \right] \nn \\
&=\!\!\!\!\prod_{{z\in \Phi \backslash \Phi_k}}\!\!\! \E_{h_z}\!\!\! \left[\exp\left(-\frac{\theta_k h_z \left|Z\right|^{-\alpha}}{\sum_{X \in \Phi_k }\left| X\right| ^{-\alpha}}\right)\!\! \Big |  \Phi \right] \nn\\
&=\prod_{n\neq k }\left(\prod_{{z\in \Phi_n}}\frac{1}{1+\frac{\theta_k \left|Z\right|^{-\alpha}}{\sum_{X \in \Phi_k }\left| X\right| ^{-\alpha}}}\right)\nn\\
& \leq\prod_{n \neq k }\frac{1}{1+\frac{\theta_k \sum_{z\in \Phi_n}\left|Z\right|^{-\alpha}}{\sum_{X \in \Phi_k }\left| X\right| ^{-\alpha}}} .\label{eq:R_1}
\end{align}
Similarly, for $\SIR$ and $\SIR''$, we have 
\begin{align}
&\Pr(\SIR>\theta_k\mid \Phi)\nn\\
&=\E_{h_n}\left[\exp\left(-\frac{\theta_k \sum_{n \neq k} h _n\sum_{Z\in \Phi_n}\left|Z\right|^{-\alpha}} { \sum_{X \in \Phi_k }\left| X\right| ^{-\alpha}}\right) \!\! \Big | \Phi \right]  \nn \\
&=\prod_{n \neq k}\E_{h_n}\left[\exp\left(-\frac{\theta_k  h _n\sum_{Z\in \Phi_n}\left|Z\right|^{-\alpha}} { \sum_{X \in \Phi_k }\left| X\right| ^{-\alpha}}\right)\!\! \Big |  \Phi \right] \nn \\
&=\prod_{n \neq k }\frac{1}{1+\frac{\theta_k \sum_{z\in \Phi_n}\left|Z\right|^{-\alpha}}{\sum_{X \in \Phi_k }\left| X\right| ^{-\alpha}}}\leq\frac{1}{1+\frac{\theta_k \sum_{z\in \Phi \backslash \Phi_k}\left|Z\right|^{-\alpha}}{\sum_{X \in \Phi_k }\left| X\right| ^{-\alpha}}},\label{eq:R_2} \\
&\Pr(\SIR''>\theta_k\mid \Phi)\nn\\
&=\E_{h_{k'}}\left[\exp\left(-\frac{\theta_k h_{k'}\sum_{z\in \Phi \backslash \Phi_k} \left|Z\right|^{-\alpha}}{\sum_{X \in \Phi_k }\left| X\right| ^{-\alpha}}\right)\!\! \Big | \Phi \right] \nn\\
&=\frac{1}{1+\frac{\theta_k \sum_{z\in \Phi \backslash \Phi_k} \left|Z\right|^{-\alpha}}{\sum_{X \in \Phi_k }\left| X\right| ^{-\alpha}}}. \label{eq:R_3}
\end{align}

Comparing the results given in \eqref{eq:R_1}, \eqref{eq:R_2} and \eqref{eq:R_3} gives the relation in the lemma statement.

\subsection{ Proof of Lemma~\ref{Lem:Mod:SIR}} \label{app:Mod:SIR}

\begin{itemize}
\item[1)] \emph{Proof for $\SIR'$}: The distribution function of $\SIR'$ can be obtained as a function of shot-noise ratio without fading, as shown in the lemma below,  so as to leverage relevant results. 

\begin{lemma}\label{lemma: Effect of interference fading}\emph{The distribution function of $\SIR'$ can be written in terms of shot-noise ratio without fading as follows: 
\begin{align}
\Pr\!\left(\SIR'\!>\!x\right)\!=\! \Pr\!\left(\!\!R\l(a_k \lambda, (1-a_k)\lambda\r)\! >\! x \frac{{\Gamma\!\!\left(1+\frac{2}{\alpha}\right)}^{\frac{\alpha}{2}}}{h_k}\!\right)\!.
\end{align}
}
\end{lemma}
\begin{proof}
According to the definitions of $\SIR'$ and $R(\lambda_1, \lambda_2)$, leveraging exponential distribution $h_k$, the distribution functions in the lemma statement can be derived using Campbell's Theorem as follows: 
\begin{align}
&\Pr\left(\SIR'>x \right)\nn\\
&=\E_{\Phi_k}\left[\exp\left(-\frac{\pi\lambda(1-a_k)x^{\frac{2}{\alpha}}}{{\left({\sum_{X \in \Phi_k }\left| X\right| ^{-\alpha}}\right)}^{\frac{2}{\alpha}}}\Gamma\l(1-\frac{2}{\alpha}\r)\right)\r] \nn\\
&\Pr\left(R(a_k \lambda, (1-a_k)\lambda) >\frac{x}{h_k}\right)\nn\\
&=\E_{\Phi_k}\left[\exp\left(-\frac{\pi\lambda(1-a_k)x^{\frac{2}{\alpha}}}{{\left({\sum_{X \in \Phi_k }\left| X\right| ^{-\alpha}}\right)}^{\frac{2}{\alpha}}}\frac{2\pi}{\alpha}\csc\frac{2\pi}{\alpha}\right)\r]. \nn
\end{align}
Comparing the right-hand sides of the above two equations reveals that they differ only in the scaling factor of $x$, which yields the desired result. 
\end{proof}

Combining the results in Lemma~\ref{lemma: Effect of interference fading} and  Proposition~\ref{pro: CCDF ratio of shot noise}, we can obtain the desired result.

\item[2)] \emph{Proof for $\SIR''$}:  {By using the Laplace function of $R(\lambda_1,\lambda_2)$ given in Proposition~\ref{pro: LT ratio of shot noise process}, 
\begin{align*}
&\Pr\l(\SIR''>x\r)\nn\\
&=\E\l[\mathcal{L}_{R\l(1-a_k,a_k\r)}\l(x h_k'\r) \r]\\
&=\sum_{m=1}^{\infty}\frac{(-1)^{m+1}}{\Gamma(1-m\frac{2}{\alpha})}  \left(\frac{a_k}{1-a_k}x^{-\frac{2}{\alpha}}\right)^{m}\E \left[{h_k'}^{-\frac{2}{\alpha}m}\right] \\
&\overset{\left( a \right)}{=}\sum_{m=1}^{\infty}\frac{(-1)^{m+1}}{\Gamma(1-m\frac{2}{\alpha})}\left(\frac{a_k}{1-a_k}x^{-\frac{2}{\alpha}}\right)^{m}\Gamma\left(1-m\frac{2}{\alpha}\right)\nn\\
& \overset{\left( b \right)}{=}\frac{1}{1+x^{\frac{2}{\alpha}} \left(\frac{1}{a_k}-1\right)} 
\end{align*}
where $(a)$ follows from $\Gamma(z)=\int_{0}^{\infty} x^{z-1}e^{-x} dx$ for all complex numbers $z$ except the non-positive integers and $(b)$ holds since the summation is a geometric series with $\frac{a_k}{1-a_k}x^{-\frac{2}{\alpha}}<1$.  Similar derivation holds for $\frac{a_k}{1-a_k}x^{-\frac{2}{\alpha}}>1$ since $\Pr\l(\SIR''>\theta_k\r)=1-\E\l[\mathcal{L}_{R\l(a_k,1-a_k\r)}\l(\frac{ h_k}{x}\r)\r]$, yielding   the same result as for the other case. This   completes   the proof. }

\end{itemize}

\subsection{Proof of Corollary~\ref{Co: special case:2}} \label{app: special case:2}
\begin{itemize}
\item[1)] \emph{Upper Bound}: Based on the inequality in Lemma~\ref{lemma: Relationships of different forms shot-noise ratio} and given $\alpha = 4$, replacing $\SIR$ with $\SIR''$ yields 
\begin{align*}
{P}_{d}(\mathcal{F}_k)&\leq\Pr\left(\frac{h_k }{h_k' \left(\frac{1-a_k}{a_k}\right)^{2}}\bar{R}>\theta_k\right)\nn\\ 
&=\E\Big[e^{-h_k'\theta_k   \left(\frac{1-a_k}{a_k}\right)^2{\bar{R}}^{-1}}\Big]\nn\\
&=\E\Bigg[\frac{1}{1+\theta_k \left(\frac{1-a_k}{a_k}\right)^{2}{\bar{R}}^{-1}}\Bigg]. 
\end{align*}
For ease of notation,  define  $v=\theta_k \left(\frac{1-a_k}{a_k}\right)^{2}$. Since the expectation in the last expression can be derived using  the CCDF of ${\bar{R}}^{-1}$, applying the result in Proposition~\ref{pro: CCDF ratio of shot noise} gives 
\begin{align*}
{P}_{d}(\mathcal{F}_k)\!\!&\leq 1\!\!-\!\!\int_{0}^{\infty}\!\!\! \frac{v}{\left(1+v x\right)^2}\left[\frac{1}{2}\!\!+\!\!\frac{2}{\pi}\arctan\!\!\left(\frac{1- \sqrt{x}}{1+\sqrt{x}}\right)\!\!\right] \!\! dx\\
&=1-\int_{0}^{\infty}\frac{v}{\left(1+v x\right)^2}\left(1-\frac{2}{\pi}\arctan\sqrt{x}\right) \!\!dx\\
&\overset{\left( a \right)}{=}\frac{4}{\pi v}\int_{0}^{\infty}\frac{y}{\left(\frac{1}{v}+ y^2\right)^2} \arctan{y}dy\nn\\
&\overset{\left( b \right)}{=}\frac{1}{1+\sqrt{v}}=\frac{1}{1+\sqrt{\theta_k} \(\frac{1}{a_k}-1\r)}
\end{align*}
where $(a)$ is by substituting $\sqrt{x}\rightarrow y$ and $(b)$ is obtained using  the the following formula in  \cite[BI (252)(12)a]{jeffrey2007table}
\begin{align*}
\int_0^{\infty}\frac{x\arctan qx}{\left(p^2+x^2\right)^2}dx=\frac{\pi q}{4p(1+pq)}.
\end{align*}

\item[2)] \emph{Lower Bound A}: For ease of notation, define  $\zeta_k = \frac{\pi\theta_k}{4}\left(\frac{1-a_k}{a_k}\right)^2 $. Again, based on the inequality in Lemma~\ref{lemma: Relationships of different forms shot-noise ratio} and given $\alpha = 4$, replacing $\SIR$ with $\SIR'$ leads to 
\begin{align*}
{P}_{d}(\mathcal{F}_k)&\geq1-\frac{2}{\pi}\E_{h_k}\left[\arctan\left(\sqrt{\frac{\zeta_k}{h_k}}\right)\right]\nn\\
&= 1-\frac{2}{\pi}\int_{0}^{\infty}\arctan\left(\sqrt{\zeta_k} x^{-\frac{1}{2}}\right)e^{-x}dx\nn\\ 
&\overset{\left( a \right)}{=}\!\! 1-\frac{2}{\pi}\left( \!\!-\frac{\sqrt{\zeta_k}}{2}\int_{0}^{\infty}\!\!\frac{x^{-\frac{3}{2}}}{1+\zeta_k x^{-1}}e^{-x}dx+\frac{\pi}{2}\!\right)\nn\\
&=\frac{\sqrt{\zeta_k}}{\pi}\int_{0}^{\infty}\frac{x^{-\frac{1}{2}}}{x+\zeta_k}e^{-x}dx
\end{align*}
where $(a)$ applies integration by parts.  Using the following formula in  \cite[EH II 137(3)]{jeffrey2007table} yields the desired result from last expression
\[
\int_0^{\infty}\frac{x^{v-1}e^{-\mu x}}{x+\beta}dx=\beta^{v-1}e^{\beta \mu}\Gamma\left(v\right)\Gamma\left(1-v,\beta\mu\right).
\]

\item[3)] \emph{Lower Bound B}: Using the fact that $\arctan$ is a concave function for $x>0$ and applying Jensen's inequality, the delivery  probability in \eqref{eq: Coverage special case} is bounded as
\begin{align}
{P}_{d}(\mathcal{F}_k)& \geq 1-\frac{2}{\pi}\arctan\left(\E\left[\frac{\sqrt{\theta_k}\sum_{n \neq k}a_n\sqrt{h_n}}{a_k\sqrt{h_k}}\right]\right) \nn \\
&\overset{\left( a \right)}{=}1-\frac{2}{\pi}\arctan\left(\frac{\pi\sqrt{\theta_k}}{2}\frac{1-a_k}{a_k}\right)
\end{align}
where $(a)$ is obtained by using  the fact that $h_k$ and $h_n$ are independent and leveraging the following equalities: 
\[\int_{x>0} x^{1/2} e^{-x}dx=\sqrt{\pi}/2,\quad  \int_{x>0} x^{-1/2} e^{-x}dx=\sqrt{\pi}. 
\]

\end{itemize}

\begin{IEEEbiography}
[{\includegraphics[width=1in,clip,keepaspectratio]{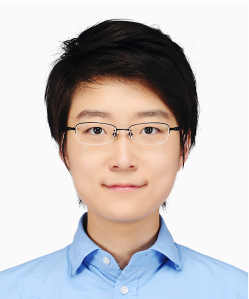}}]
{Dongzhu Liu}(S'17) received the B.Eng. from the University of Electronic Science and Technology of China (UESTC) in 2015. Since Sep. 2015, she has been a PhD student in the Dept. of Electrical and Electronic Engineering (EEE) at The University of Hong Kong. Her research interests focus on the analysis and design of wireless networks using stochastic geometry.
\end{IEEEbiography}

\begin{IEEEbiography}
[{\includegraphics[width=1in,clip,keepaspectratio]{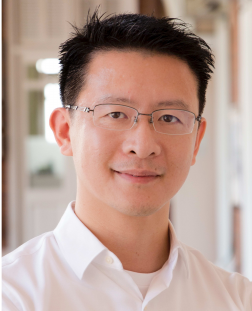}}]
{Kaibin Huang} (M'08-SM'13) received the B.Eng. (first-class hons.) and the M.Eng. from the National University of Singapore, respectively, and the Ph.D. degree from The University of Texas at Austin (UT Austin), all in electrical engineering. Since Jan. 2014, he has been an assistant professor in the Dept. of Electrical and Electronic Engineering (EEE) at The University of Hong Kong. He used to be a faculty member in the Dept. of EEE at Yonsei University in S. Korea and currently with the department as an adjunct professor. His research interests focus on the analysis and design of wireless networks using stochastic geometry and multi-antenna techniques. 

He frequently serves on the technical program committees of major IEEE conferences in wireless communications. He has been the technical chair/co-chair for the IEEE CTW 2013, the Wireless Communications Symposium of IEEE GLOBECOM 2017, the Comm. Theory Symp. of IEEE GLOBECOM 2014, and the Adv. Topics in Wireless Comm. Symp. of IEEE/CIC ICCC 2014 and has been the track chair/co-chair for IEEE PIMRC 2015, IEE VTC Spring 2013, Asilomar 2011 and IEEE WCNC 2011. Currently, he is an editor for IEEE Journal on Selected Areas in Communications (JSAC) series on Green Communications and Networking, IEEE Transactions on Wireless Communications, IEEE Wireless Communications Letters. He was also a guest editor for the JSAC special issues on communications powered by energy harvesting and an editor for IEEE/KICS Journal of Communication and Networks (2009-2015). He is an elected member of the SPCOM Technical Committee of the IEEE Signal Processing Society. Dr. Huang received the 2015 IEEE ComSoc Asia Pacific Outstanding Paper Award, Outstanding Teaching Award from Yonsei, Motorola Partnerships in Research Grant, the University Continuing Fellowship from UT Austin, and a Best Paper Award from IEEE GLOBECOM 2006.
\end{IEEEbiography}

\end{document}